\documentclass[a4paper,USenglish,cleveref,hyperref,thm-restate,numberwithinsect]{lipics-v2021}
\pdfoutput=1

\usepackage{amsmath,amsthm, amssymb}
\usepackage[table]{xcolor}
\usepackage{tikz}
\usetikzlibrary{math, positioning, decorations.pathreplacing, backgrounds, decorations.shapes, arrows.meta, calc, fit}
\definecolor{editred}{HTML}{D62828}
\definecolor{pipeblue}{HTML}{457B9D}
\definecolor{c3}{RGB}{0,128,0}
\definecolor{c4}{RGB}{255,140,0}
\definecolor{regime1}{HTML}{E74C3C}
\definecolor{regime2}{HTML}{F39C12}
\definecolor{regime3}{HTML}{27AE60}
\usepackage{stmaryrd}
\usepackage{booktabs}
\usepackage{xspace}
\usepackage[linesnumbered, ruled, vlined, noend]{algorithm2e}

\definecolor{Gray}{gray}{0.85}
\newcolumntype{a}{>{\columncolor{Gray}}r}

\newcommand{\card}[1]{\textsf{card}\left( #1 \right)}

\newcommand{\LZ}{\textsc{LZ77}\xspace}
\newcommand{\CLZ}[1]{C_{\mathrm{LZ77}}(#1)}

\newcommand{\factor}{\sqsubseteq}
\newcommand{\DeltaDist}[2]{\Delta\!\bigl(#1,#2\bigr)}
\newcommand{\nn}{\mathbb{N}}%
\newcommand{\eps}{\epsilon}%

\newcommand{\corrindex}[1]{\textcolor{black}{#1}}

\title{The \texorpdfstring{$k$}{k}-bit Sensitivity of LZ77}

\author{Gabriel Bathie}{LaBRI, University of Bordeaux, France \and \url{https://gbathie.github.io/}}{gabriel.bathie@gmail.com}{https://orcid.org/0000-0003-2400-4914}{}
\author{Paul Huber}{LaBRI, University of Bordeaux \and ENS Paris-Saclay, France }{paul.huber@ens-paris-saclay.fr}{}{}
\author{Guillaume Lagarde}{LaBRI, University of Bordeaux, France \and \url{https://guillaume-lagarde.github.io/}}{guillaume.lagarde@labri.fr}{}{}
\author{Akka Zemmari}{LaBRI, University of Bordeaux, France \and \url{https://www.labri.fr/perso/zemmari/}}{akka.zemmari@u-bordeaux.fr}{}{}
\authorrunning{G. Bathie, P. Huber, G. Lagarde and A. Zemmari}

\Copyright{Gabriel Bathie, Paul Huber, Guillaume Lagarde and Akka Zemmari}

\hideLIPIcs
\nolinenumbers

\begin{document}
\maketitle

\begin{abstract}
  In this paper, we study the sensitivity of the Lempel--Ziv 77 (LZ77) algorithm under \emph{multiple} substitutions.

  Akagi, Funakoshi and Inenaga [JIC, 2023] showed that when applying a \emph{single} edit the size $z$ of the LZ77 compression of any string increases by a factor of at most 2.
  This raises the question of what happens under multiple edits. The result of Akagi et al. implies that the compression increases by a factor of at most $2^k$.
  We conjectured that the optimal multiplicative factor for the sensitivity of LZ77 under $k$ edits was $\Theta(k)$. Rather surprisingly, we show that this sensitivity is \emph{constant}, with $k$ appearing in a linear additive term.
  Namely, we show that if $w'$ is at Hamming distance at most $k$ from $w$, we have $C_{\mathrm{LZ77}}(w') \leq 3 \cdot C_{\mathrm{LZ77}}(w) + 4k$.

  This bound is almost tight: we exhibit words for which
  $C_{\mathrm{LZ77}}(w') \ge \left(3\cdot C_{\mathrm{LZ77}}(w)+k\right)(1-\epsilon)$.
  This result contrasts with Lempel--Ziv 78, where a single edit can significantly deteriorate compressibility, a phenomenon known as a \emph{one-bit catastrophe}.

  We further refine this bound, focusing on the coefficient $3$ in front of $C_{\mathrm{LZ77}}(w)$, and establish a trichotomy based on the compressibility of $w$. More precisely, for $n \gg k^{7}$, we prove the following bounds:
  \begin{itemize}
    \item if $C_{\mathrm{LZ77}}(w) \lesssim k^{3/2}\sqrt{n}$, the compression may increase by up to a factor of $\approx 3$,
    \item if $k^{3/2}\sqrt{n} \lesssim C_{\mathrm{LZ77}}(w) \lesssim k^{1/3}n^{2/3}$, this factor is at most $\approx 2$,
    \item if $C_{\mathrm{LZ77}}(w) \gtrsim k^{1/3}n^{2/3}$, the factor is at most $\approx 1$.
  \end{itemize}
\end{abstract}

\section{Introduction}\label{sec:intro}

The starting point of this work lies in recent studies on the sensitivity of lossless data compression algorithms against small perturbations.  This line of research goes back at least to the so-called \emph{one-bit catastrophe} conjecture for the Lempel--Ziv 78 algorithm (LZ78)~\cite{ZL78}, which was originally raised by Jack Lutz and others in the late 1990s and explicitly stated for instance in~\cite{LathropS97,PierceS2000, Lopez06}. The conjecture asked whether adding a single bit at the beginning of a compressible string could transform it into an incompressible one. Surprisingly, such a phenomenon was proved to be possible by Lagarde and Perifel~\cite{LagardeP18}.

Since then, the robustness of various lossless compression algorithms has been investigated under single bit modifications, insertions, or deletions. For instance, Giuliani \emph{et al.}~\cite{Giuliani2025} demonstrate analogous \emph{bit catastrophes} for the Burrows–Wheeler transform, Nakashima \emph{et al.}~\cite{nakashima24} analyze the sensitivity of the lex-parse compressor, and Akagi, Funakoshi and Inenaga~\cite{AkagiFI23} systematically analyze the sensitivity of several string repetitiveness measures and compression algorithms, including substring complexity, \LZ~\cite{ZL77}, LZ78, LZ-End, and grammar-based compressors such as AVL grammars~\cite{Rytter03}. Their work formalizes \emph{sensitivity} as the worst-case change in the compressed size, or repetitiveness measure, caused by a one-bit edit, providing both upper and lower bounds for most measures.

Prior to the formal establishment of ``catastrophic'' behaviors, the compression community had already observed the algorithm’s instability in practice and had begun developing techniques to mitigate its sensitivity. A strong line of work involves relaxing the rigid greedy parsing strategy of Ziv--Lempel algorithms by trying different starting points at each stage of compression. These approaches, often referred to as \emph{greedy lookahead}, incorporate local search methods to improve global compression performance.
One early instance of this idea is the \emph{flexible parsing} algorithm of Matias and Sahinalp~\cite{MatiasRS99}, which augments LZ78 with a one-step lookahead to select among multiple candidate phrase starts. Experimental results showed substantial gains in compression, up to 35\% improvements on DNA and medical data, without significant computational overhead~\cite{MatiasRS99,MatiasWAE}. Subsequent work by Arroyuelo \emph{et al.}~\cite{AronicaLMMM18} and Langiu~\cite{OnOptimalParsing} investigated the optimality of such greedy approaches. They showed that flexible parsing is in fact \emph{optimal} under constant-cost phrase encodings, and \emph{almost optimal} in the more realistic setting where the cost of a phrase depends on its dictionary index or bit-length.

The idea of transforming a string before compression in order to improve its LZ77 compressibility is also not entirely new. Bannai, Charalampopoulos and Radoszewski \cite{BCR24} study the maintenance of the LZ77 parse size under semi-dynamic updates and derive, as a corollary, an algorithm for computing a most LZ77-compressible rotation of a string. This provides a first indication that modifying the input prior to compression may be beneficial. Our perspective here is different: rather than focusing on a specific transformation such as cyclic rotation, we study the general effect of k edits on the LZ77 parse size, both as a source of deterioration and as a possible lever for improvement.

Among all lossless compression algorithms, the Ziv-Lempel family of compressors is undoubtedly one of the most studied and widely used in both theory and practice, making it a natural starting point for our investigation. We believe that extending these questions to broader classes of compressors is an important direction for future work, both to deepen the theoretical understanding of compression algorithms and to potentially inspire improved algorithms in practice. The Ziv–Lempel algorithms have had a profound theoretical influence: the optimality of LZ among finite-state compressors, proved by Ziv and Lempel, initiated a rich line of research connecting data compression to multiple areas of theoretical computer science. These include, to name a few, information theory, through Shannon entropy~\cite{ZL77, ZL78}, number theory, via the construction of absolutely normal numbers~\cite{LM16}, and computational complexity, through Hausdorff dimensions of complexity classes introduced by Lutz~\cite{Lutz03}, in particular finite-state and polynomial-time dimensions~\cite{Lopez05}. We focus here on \LZ, for two main reasons. First, the combinatorial structure of \LZ parsing is comparatively simple and well-behaved, in contrast to other schemes from the Ziv-Lempel family such as LZ78, whose parsing can exhibit a more chaotic behavior under small perturbations. This structural regularity makes \LZ particularly well suited as a first case study for understanding how multiple edits affect compression. Second, \LZ enjoys widespread use in real-world applications. For instance, it underlies the DEFLATE algorithm in tools like gzip and PNG compression, making it relevant from a practical point of view.

\subsection{Our Results}

\subsubsection{Tight Upper and Lower Bounds in the General Case}

Our first result is a generalization of the sensitivity analysis of \LZ compression by Akagi, Funakoshi and Inenaga~\cite{AkagiFI23}, which considers a single edit, to the general case of $k$ edits. For a single edit, Akagi et al. proved that the number of phrases increases by a factor of at most $2$, a bound they also showed to be tight. Iterating this bound naively yields an exponential upper bound of $\CLZ{w'} \le 2^k \cdot \CLZ{w}$. However, a closer inspection suggests a tighter bound. Indeed, a simple observation is that $k$ edits can fragment any single phrase into at most $k+1$ pieces. Therefore one might naturally conjecture a linear bound of $\CLZ{w'} \le (k+1)\cdot \CLZ{w}$. Somewhat surprisingly, we prove that the multiplicative blow-up does not grow with $k$. The number of phrases increases by a factor of at most $3$, with the dependence on $k$ isolated into an additive term. This is the purpose of \cref{thm:upper-general}.

Here and throughout, $B(w,k)$ denotes the Hamming ball of radius $k$ around $w$, that is, the set of
words obtained from $w$ by substituting \emph{at most} $k$ positions (see \cref{sec:preliminaries}).

\begin{restatable*}{theorem}{UpperGeneral}
  \label{thm:upper-general}
  For all $k \in \mathbb{N}$, $w \in \Sigma^*$ and $w' \in B(w,k)$, we have
  \[
    \CLZ{w'} \le 3\cdot \CLZ{w} + 4k.
  \]
\end{restatable*}

\begin{remark}
  The fact that the dependence in $k$ is \emph{additive} and entirely decoupled from $\CLZ{w}$  shows that \LZ is robust to perturbations. No matter how large $k$ is, the multiplicative degradation does not exceed $3$ unless $k$ is very large, i.e., of order $C(w)$.
\end{remark}

Additionally, we prove in \cref{prop:lower-general} that the bound of \cref{thm:upper-general} is essentially tight up to an additive term linear in~$k$.

\begin{restatable*}{proposition}{LowerGeneral}
  \label{prop:lower-general}
  For all $0 < \epsilon \le 1$, there exists infinitely many $k \in \mathbb{N}$ for which there exist words $w \in \Sigma^*$ and $w' \in B(w,k)$ such that
  \[
    \CLZ{w'} \ge \left(3\cdot \CLZ{w} + k\right) (1 - \epsilon).
  \]
\end{restatable*}

\subsubsection{Tight Upper and Lower Bounds Depending on the Number of Edits and Compression Ratio.}

We then further refine this analysis and show that the multiplicative sensitivity of $\CLZ{w}$ depends on the compressibility of $w$. We show in \cref{cor:trichotomy} that there are three distinct regimes with multiplicative constants of 1, 2 and 3 respectively, with sharp phase transitions between them.
A representation of this situation is given in \Cref{fig:tricho_asympt}.

\begin{restatable*}{corollary}{CorTrichotomy}\label{cor:trichotomy}
  Let $w$ be a word of length $n$, let $k\le n$ be the number of edits, let $w'\in B(w, k)$,
  and let $r= \frac{\CLZ{w'}}{\CLZ{w}}$.
  Let $0 < \delta < 1$ and assume that the compressed size is large compared to the number of
  edits, i.e. that we have $\CLZ{w} \ge \frac{18 k}{\delta}$.
  Then:
  \[
    \begin{cases}
      \text{if } \CLZ{w} > \dfrac{12}{\delta}\max\bigl\{k^{1/3}n^{2/3},\ k^{3/2}\sqrt n\bigr\}, \text{then} & r < 1 + \delta,\\[2mm]
      \text{if } \CLZ{w} > \dfrac{2\,k^{3/2}\sqrt{n}}{\delta}, \text{then} & r < 2 + \delta,\\[2mm]
      \text{else } & r \le 3\corrindex{+\delta}.
    \end{cases}
  \]
\end{restatable*}

\begin{remark}
  Let us discuss the need for the hypothesis that $\CLZ w\ge 18k/\delta$. The proof of \cref{cor:trichotomy} rests on the bound
  \[
    r  \le 1 + p_2(w,w') + p_3(w,w') + \frac{4k}{\CLZ w},
  \]
  and the additive term $4k/\CLZ w$, together with the
  additive term $k$ of \cref{thm:upper-regime}, is negligible only when
  $\CLZ w\gg k$; this is what this hypothesis expresses.
\end{remark}

\begin{figure}[htbp]
  \centering
  \begin{tikzpicture}[scale=1.5, every node/.style={font=\small}]
    \draw[->, thick] (-0.3,0) -- (6.8,0) node[right] {$\CLZ{w}$};
    \draw[->, thick] (0,-0.3) -- (0,4) node[above] {$r_{\max}$};

    \foreach \y in {1,2,3} {
      \draw[thick] (-0.1,\y) -- (0.1,\y);
      \node[left] at (-0.1,\y) {\y};
    }

    \fill[regime1, opacity=0.15] (0,0) rectangle (2,3);
    \draw[very thick, regime1] (0,3) -- (2,3);
    \draw[very thick, regime1] (2,3) -- (2,2);
    \draw[dashed, gray] (2,0) -- (2,2);
    \node[below] at (2,-0.1) {$\approx k^{3/2}\sqrt{n}$};
    \node[regime1, font=\footnotesize\bfseries, align=center] at (1,3.3) {Highly\\compressible};

    \fill[regime2, opacity=0.15] (2,0) rectangle (4,2);
    \draw[very thick, regime2] (2,2) -- (4,2);
    \draw[very thick, regime2] (4,2) -- (4,1);
    \draw[dashed, gray] (4,0) -- (4,1);
    \node[below] at (4,-0.1) {$\approx k^{1/3}n^{2/3}$};
    \node[regime2, font=\footnotesize\bfseries, align=center] at (3,2.3) {Moderately\\compressible};

    \fill[regime3, opacity=0.15] (4,0) rectangle (6,1);
    \draw[very thick, regime3] (4,1) -- (6,1);
    \draw[very thick, regime3] (6,1) -- (6,0);
    \draw[dashed, gray] (6,0) -- (6,1);
    \node[below] at (6,-0.1) {$\frac{n}{\log n}$};
    \node[regime3, font=\footnotesize\bfseries, align=center] at (5,1.3) {Poorly\\ compressible};

  \end{tikzpicture}
  \caption{Representation of the sensitivity trichotomy when $n\to\infty$. This plot shows $r_{\max}$, the maximum value of the growth ratio $r=\frac{\CLZ{w'}}{\CLZ{w}}$ of the compression of $w$ under $k$ edits, as a function of the compressibility of $w$.
    \corrindex{The picture assumes $k^{3/2}\sqrt n\le k^{1/3}n^{2/3}$, i.e. that $n\ge k^{7}$ (see \cref{cor:trichotomy}). The right end of the axis is $n/\log n$ because the \LZ parse of a word of length $n$ over a fixed alphabet always has $O(n/\log n)$ phrases~\cite{ZL77}.}
  }
  \label{fig:tricho_asympt}
\end{figure}

\begin{remark}
  \cref{cor:trichotomy} states an intuitive phenomenon: the more compressible a word is, the more sensitive its compression can be to small perturbations. However, one might have expected a continuous spectrum rather than only three distinct regimes.
\end{remark}

\cref{cor:trichotomy} is based on \cref{thm:upper-regime}, which analyzes how the parse of $w'$ aligns with that of $w$. Let $p_2(w, w')$ and $p_3(w, w')$ denote the fraction of LZ blocks in $w$ that contain at least two and at least three phrase-starts from $\mathsf{LZ77}(w')$, respectively (formal definitions are given in \cref{sec:preliminaries}). We prove the following:

\begin{restatable*}{theorem}{UpperRegime}
  \label{thm:upper-regime}
  For all $k \in \mathbb{N}$, $w \in \Sigma^*$ of length $n$ and $w' \in B(w,k)$, we have
  \[
    p_2(w,w')\cdot\CLZ{w} \le 4 k^{1/3}n^{2/3} \corrindex{+\;k}
    \text{ and }
    p_3(w,w')\cdot\CLZ{w} \le \frac{k^{3/2} \sqrt{n}}{\sqrt2} \corrindex{+\;k},
  \]
  or equivalently
  \[
    \CLZ{w} \le \frac{4 k^{1/3}n^{2/3}+k}{p_2(w,w')} \text{ and }
    \CLZ{w} \le \frac{1}{p_3(w,w')}\left(\frac{k^{3/2} \sqrt{n}}{\sqrt2}+k\right).
  \]
\end{restatable*}










Finally, we show that the sub-regimes presented in \cref{cor:trichotomy} are almost tight.

\begin{restatable*}{proposition}{TrichotomyLower}
  \label{prop:trichotomy-lower}
  Let $\delta > 0$ and $m \in \mathbb{N}$. The following two statements hold.

  \begin{enumerate}
    \item
      For every $k\ge1$, there exist $n \ge m$ and words $w, w'$ of length $n$ with $w' \in B(w,k)$ such that
      \[
        \CLZ{w} \ge {k^{1/3} n^{2/3}}(1-\delta)
        \text{ and }
        \frac{\CLZ{w'}}{\CLZ{w}} \ge 2 - \delta.
      \]

    \item
      For every \corrindex{even $k\ge 2$}, there exist $n \ge m$ and words $w, w'$ of length $n$ with $w' \in B(w,k)$ such that
      \[
        \CLZ{w} \ge \sqrt{\frac{kn}{2}}(1-\delta)
        \text{ and }
        \frac{\CLZ{w'}}{\CLZ{w}} \ge 3 - \delta.
      \]
  \end{enumerate}
\end{restatable*}

\begin{remark}
  The parity constraint in the second statement is necessary and not an artefact of the construction,
  which spends two substitutions per block ($k=2k'$). Indeed, for $k=1$ the sensitivity of \LZ is at
  most a factor of $2$ by~\cite{AkagiFI23}, so no family of words can achieve a ratio $3-\delta$ with a
  single edit. For odd $k\ge3$, applying the construction with $k-1$ edits gives words with
  $\CLZ{w}\ge\sqrt{(k-1)n/2}\,(1-\delta)$ and the same ratio $3-\delta$.
\end{remark}

\begin{remark}
  \cref{prop:trichotomy-lower} confirms that the first regime in
  \cref{cor:trichotomy} is tight up to constant factors.
  In the second regime, however, the lower bound falls short by a multiplicative
  factor of $k$: our construction achieves $\CLZ{w} = \Omega(\sqrt{kn})$, whereas
  the upper bound shows $O(k^{3/2}\!\sqrt{n})$.
  We conjecture that the upper bound is asymptotically tight but determining this remains an open problem.
\end{remark}

\subsection{Organization of the Paper and Technical Overview.}
We now give a more detailed overview of the ideas behind our results and outline the organization of the paper.

\cref{sec:preliminaries} introduces the necessary notation and definitions.

In \cref{sec:lz77-sensitivity}, we present the proof of \cref{thm:upper-general} and \cref{prop:lower-general} that provide the general upper and lower bounds on the sensitivity of \LZ under $k$ edits. Unlike the single-edit case, where each phrase in $w$ may be split into at most two parts, multiple edits can fragment phrases of $w$ into several smaller segments, creating intricate dependencies between splits. To analyze this fragmentation, we introduce a \emph{jump function} (\cref{def:jump-function}), which tracks how phrases in $w$ are progressively reconstructed in the edited word $w'$. Intuitively, the jump function describes how far one can ``skip'' forward in $w'$ before crossing another edit, and its \emph{jump depth} measures the cumulative effect of these crossings. By bounding the total jump depth, we obtain the desired upper bound on $\CLZ{w'}$. The lower bound is proved by constructing a family of words that realize the worst-case pattern, showing that the upper bound is essentially tight.

\cref{sec:regimes} is devoted to \cref{thm:upper-regime}, \cref{cor:trichotomy} and \cref{prop:trichotomy-lower}, which together yield the trichotomy of sensitivity regimes.
To prove \cref{thm:upper-regime}, we observe that the set of phrases in $w$ that generate at least two or three new phrase-starts in $w'$, \corrindex{after discarding the at most $k$ phrases that carry an edited position}, possesses a well-defined combinatorial structure.
We capture this structure via two objects, the \emph{$P_2$} and \emph{$P_3$-systems} (Definitions~\ref{def:p2} and~\ref{def:p3}), which provide a compact representation of the dependencies created by local modifications.
This representation reduces the analysis to an extremal combinatorics problem.
The main technical step is to establish lower bounds for these systems (Propositions~\ref{prop:p2} and~\ref{prop:p3}), which then directly imply the bounds in \cref{thm:upper-regime}.
We solve these extremal problems using optimization under constraints techniques.
The $P$-systems, or suitable variants, may be of independent interest both for the analysis of other dictionary-based compressors and from a purely combinatorial viewpoint.

\section{Preliminaries}\label{sec:preliminaries}

\subsection{Notation}
Let $\Sigma$ be a finite alphabet. We write $\Sigma^*$ (resp. $\Sigma^n$) for the set
of all finite strings (resp. all strings of size $n$) over $\Sigma$. We also call a string a
word. For $w\in\Sigma^*$, let $\lvert w\rvert$ denote its length,
and for $1 \le i \le j \le \lvert w\rvert$, we write $w[i:j] \;=\; w_i w_{i+1}\cdots w_j$
for the substring of $w$ from position $i$ to $j$.  We use $\varepsilon$ for the empty string.
We write $u \factor w$ to indicate that $u$ is a substring of $w$ (i.e., a contiguous subsequence).
For two strings $w,w'\in\Sigma^n$, their \emph{Hamming distance} is
\[
  \DeltaDist{w}{w'}
  \;=\;
  \card{\{i : w_i \neq w'_i\}}.
\]
For $k\in\mathbb{N}$ we define the Hamming ball of radius $k$ around $w$ by
\[
  B(w,k)
  \;=\;
  \bigl\{w'\in\Sigma^{\lvert w\rvert} : \DeltaDist{w}{w'}\le k\bigr\}.
\]
Note that $B(w,k)$ is defined with $\le k$, so that ``$k$ edits'' always means ``\emph{at most} $k$
substitutions'' throughout the paper. All our upper bounds are monotone in $k$, hence hold under both
readings; in the lower bounds, the constructions use exactly $k$ substitutions.
Finally, a \emph{proper prefix} of a string $u$ is a prefix of $u$ distinct from $u$ itself.

\subsection{\LZ parsing}
Let $w = w_1w_2\cdots w_n$ be a word of length $n$. The \LZ \emph{factorization} of $w$ is a decomposition $w \;=\; P_1P_2\cdots P_z$ computed greedily from left to right as follows. \corrindex{The construction starts with $i=0$ and $m=0$, that is, with the empty parsed prefix.} Suppose we have already parsed the prefix $w[1\colon m]$ into $P_1\ldots P_i$, and let $R = w[m+1\colon n]$ denote the yet-unparsed part of $w$.
To construct $P_{i+1}$, let $S$ be the longest prefix of $R$ that appears in $w$ at some position $j\le m$, \corrindex{with $S$ the empty word when no non-empty prefix of $R$ does --- which is always the case for $P_1$, since $m=0$ leaves no position $j\le m$ available}.
If $S \ne R$, we let $P_{i+1} = R[1:|S|+1]$, i.e. $P_{i+1}$ is equal to $S$ plus one more letter of $R$, otherwise $P_{i+1} = S = R$, and the algorithm is over.
Note that this definition allows the matched $S$ substring to overlap with the yet-unparsed suffix of $w$ (i.e., the match can extend beyond position $m$). This variant of \LZ is sometimes referred to as \LZ with self-reference. For simplicity, we will simply call it \LZ in the rest of the paper.

Each $P_i$ is called a \emph{phrase}. Each phrase is encoded as a tuple representing the starting index of the match $S$ in $w$, its length, and the newly added character (see Figure~\ref{fig:parsing_expl} for an example). Notice that if $S$ is empty, the phrase is the single letter $P_{i+1} = w_{m+1}$ (encoded as $(0, 0, w_{m+1})$); \corrindex{this is in particular how the parse starts, with $P_1=w_1$}. We then increase $m$ by $\lvert P_{i+1}\rvert$ and repeat until $m=n$.

We use $\LZ(w) = (P_1, \ldots, P_z)$ to denote the list of phrases in the \LZ factorization of $w$. The number of phrases in $\LZ(w)$ is denoted by $\CLZ{w}$.
We denote by $l(P_i)$ and $r(P_i)$ the starting and ending positions of the phrase $P_i$ in $w$, i.e. we have $P_i = w[l(P_i): r(P_i)]$.

\begin{figure}[ht]
  \begin{center}
    \begin{tikzpicture}[x=0.5cm, y=0.6cm]

      \node[anchor=east, align=right] at (-1.5, 0) {\LZ : };

      \definecolor{c1}{RGB}{220,20,60}   
      \definecolor{c2}{RGB}{30,144,255}  
      \definecolor{c3}{RGB}{0,128,0}     
      \definecolor{c4}{RGB}{255,140,0}   
      \definecolor{c5}{RGB}{128,0,128}   

      \node[text=c1] at (0,0) {a};
      \node[text=c2] at (1,0) {b};
      \node[text=c3] at (2,0) {b};
      \node[text=c3] at (3,0) {b};
      \node[text=c3] at (4,0) {a};
      \node[text=c4] at (5,0) {b};
      \node[text=c4] at (6,0) {a};
      \node[text=c4] at (7,0) {b};
      \node[text=c4] at (8,0) {a};
      \node[text=c4] at (9,0) {b};
      \node[text=c4] at (10,0) {a};
      \node[text=c4] at (11,0) {b};
      \node[text=c4] at (12,0) {b};
      \node[text=c5] at (13,0) {a};
      \node[text=c5] at (14,0) {a};

      \draw[c1, thick] (0.5,-0.3) -- (0.5,0.3);
      \draw[c2, thick] (1.5,-0.3) -- (1.5,0.3);
      \draw[c3, thick] (4.5,-0.3) -- (4.5,0.3);
      \draw[c4, thick] (12.5,-0.3) -- (12.5,0.3);

      \node[text=c1] at (0,-0.7) {\scriptsize$(0,0, a)$};
      \node[text=c2] at (1,0.7) {\scriptsize$(0,0,b)$};
      \node[text=c3] at (3,-0.7) {\scriptsize$(2,2,a)$};
      \node[text=c4] at (8.5,0.7) {\scriptsize$(4,7,b)$};
      \node[text=c5] at (13.5,-0.7) {\scriptsize$(1,1,a)$};

    \end{tikzpicture}
  \end{center}
  \caption{
    An example of \LZ parsing on the word \texttt{abbbababababbaa}. Between parentheses are the results of compressions of each block with, in order, the starting position of the referenced substring, its length and the letter added. Note that the reference for the fourth phrase (starting at index 4 with length 7) overlaps with the phrase itself.
  }
  \label{fig:parsing_expl}
\end{figure}

\subsection{Alignment parameters.}
Let $w = P_1P_2\cdots P_z$ and $w' = Q_1Q_2\cdots Q_{z'}$ be \LZ factorizations of $w$ and $w'$.
For a phrase $P_i$, let $\textsf{intersect}(P_i) = \card{\{j : l(P_i)\le l(Q_j)\le r(P_i)\}}$ denote the number of phrases $Q_i$ that start at a position contained in $P_i$.
Define
\[
  p_2(w,w')
  = \frac{1}{z} \card{\left\{i : \textsf{intersect}(P_i)\ge 2\right\}} \text{ and }
  p_3(w,w')
  = \frac{1}{z} \card{\left\{i : \textsf{intersect}(P_i)\ge 3\right\}}.
\]
These give the fraction of original phrases containing at least two or at least three new phrase-starts.
When $w$ and $w'$ are clear from the context, we will simply write $p_2$ and $p_3$ for $p_2(w,w')$ and $p_3(w,w')$.

\section{General upper and lower bounds on \LZ sensitivity}\label{sec:lz77-sensitivity}

\subsection{A general upper bound on \LZ sensitivity}

In this section, we generalize the sensitivity analysis of \LZ compression by Akagi, Funakoshi and Inenaga~\cite{AkagiFI23}, which considers a single edit, to the general case of $k$ edits. While a single edit can increase the number of phrases by at most a multiplicative factor of $2$, we show that multiple edits can increase the number of phrases by a factor of $3$. This increase arises because a single edit breaks at most one phrase into two, whereas multiple edits may fragment phrases into more than two parts.

Our main result is a tight upper bound on the \LZ complexity of a string $w'$ at Hamming distance $k$ from a reference string $w$.

\UpperGeneral

Before going into the full technical proof, let us sketch the main
ideas.  When we edit up to $k$ positions in a string $w$, each
edit can potentially break existing \LZ phrases into multiple new
fragments.  To control this fragmentation, we introduce a jump
function $g_j$ that, at each position
$j$, tracks how far we can ``jump'' over the next block of edits by
matching substrings in the already-parsed prefix of the edited text.
Iterating $g_j$ tells us how many edits lie between successive
recoverable matches, and thus how many extra phrases may be forced by
those edits.  We capture this with the quantity $\|g_j\|$, that denotes the least
number of jumps needed to clear all $k$ edits.  Intuitively,
$\|g_j\|$ decreases as we extend the parsed prefix past more edits,
and the difference $\|g_{p}\|-\|g_{q}\|$ measures how many
new phrases arise in the interval $[p,q]$.  Bounding these
increments over all original phrases leads directly to the stated
$3\cdot\CLZ{w}+4k$ upper bound.

We start by defining the jump function and its first properties.

\begin{definition}[Jump-map and jump-depth]\label{def:jump-function}
  Fix substitution positions
  \[
    \corrindex{1} \le i_1 < i_2 < \cdots < i_k \le \corrindex{n},
    \qquad i_0 = \corrindex{0},
    \qquad i_{k+1} = n+1 ,
  \]
  where $i_0$ and $i_{k+1}$ are two \emph{sentinels} that stand for the two ends of the word,
  and let $w' \in \Sigma^n$ be the string obtained from $w$ by replacing each
  $w[i_\ell]$ by a fresh symbol $\#_\ell$ (for $1\le\ell\le k$).
  For each $t \in \{\corrindex{1},\dots,\corrindex{n+1}\}$, define $g_t : \{0,1,\dots,k+1\} \to \{0,1,\dots,k+1\}$
  by $g_t(k+1)=k+1$, and for $0 \le \ell \le k$,
  \[
    \begin{aligned}
      g_t(\ell) = &
      \begin{cases}
        \ell+1,
        \quad\text{if } w[i_\ell+1 : i_{\ell+1}-1] \not\factor w'[\corrindex{1}:t-1],\\
        \max\Bigl\{\ell' > \ell : w[i_\ell+1 : i_{\ell'}-1] \factor w'[\corrindex{1}:t-1]\Bigr\},
        \quad\text{otherwise.}
      \end{cases}
    \end{aligned}
  \]
  The \emph{jump-depth} of $g_t$ is defined as $\|g_t\|=\min\{s : g_t^{s}(0)\ge k+1\}$.
  It is the least number of iterations of $g_t$ needed to reach the \corrindex{state $k+1$} starting from $0$.
  The sentinel state $k+1$ plays the role of ``all the substitutions have been cleared'': reaching it
  from $\ell$ means that the whole suffix $w[i_\ell+1:n]$ lying after the $\ell$-th substitution has been covered.
  Using $k+1$ rather than $k$ as the absorbing state is what allows us to treat the last substitution
  exactly like the others; it costs nothing, since it shifts every jump-depth by at most one and only
  \emph{differences} of jump-depths are used below.
\end{definition}


\begin{remark}
  The use of fresh characters to substitute at positions $i_j$ is the worst case here. Indeed, using other characters could only increase the number of factors in $w'[\corrindex{1}:i]$ used later in the compression. Thus, using non-fresh characters could only reduce the size of the compression $\CLZ{w'}$.
\end{remark}

We start with some basic properties of the jump function $g_t$ that will be useful later.
\begin{proposition}[Monotonicity of jumps]\label{rk:monotone}
  For all $\corrindex{1}\le t\le t'\le \corrindex{n+1}$, all $0\le \ell \le \ell'\le \corrindex{k+1}$ and all $s\ge 0$:
  \begin{enumerate}
    \item $g_t(\ell)>\ell$ \corrindex{when $\ell\le k$}.
    \item $g_t(\ell)\le g_{t'}(\ell)$
    \item $g_t(\ell) \le g_{t}(\ell')$
    \item\label{rk-monotone:4} $\|g_t\|\ge \|g_{t'}\|$.
    \item\label{rk-monotone:5} $g_t^{s}(\ell) \le g_{t'}^{s}(\ell)$.
  \end{enumerate}
  Items~(2) and~(3) hold at the absorbing state as well: $g_t(k+1)=k+1$ is the largest value taken by
  any $g_{t'}$, so both inequalities hold as soon as $\ell=k+1$ or $\ell'=k+1$. This is needed
  below, since the iterates of $g_t$ do reach $k+1$.
  Item~(5) follows from items (2) and (3) by induction on $s$: if
  $g_t^{s}(\ell)\le g_{t'}^{s}(\ell)$, then
  $g_t^{s+1}(\ell) = g_t\bigl(g_t^{s}(\ell)\bigr) \le g_{t}\bigl(g_{t'}^{s}(\ell)\bigr)
  \le g_{t'}\bigl(g_{t'}^{s}(\ell)\bigr) = g_{t'}^{s+1}(\ell)$.
\end{proposition}

The main technical ingredient behind the proof of \cref{thm:upper-general} is the following proposition:

\begin{proposition}\label{prop:interval-bound}
  Let $P_j$ be the $j$-th phrase of $LZ77(w)$ with endpoints $l_j=l(P_j), r_j=r(P_j)$,
  and let $n_j$ be the number of substitutions in $[l_j,r_j]$.
  Then, the number of new phrase-starts of $w'$ in $[l_j,r_j]$ satisfies
  \begin{equation}
    \card{\{i: l_j \le l(Q_i)\le r_j\}}
    \le
    3 + \|g_{l_j}\| - \|g_{r_j}\|+ 3 n_j.\label{eq:interval_bound}
  \end{equation}
\end{proposition}
We defer the proof of \cref{prop:interval-bound} to the appendix, in \cref{app:jump-depth}.

With \cref{prop:interval-bound} in hand, we can directly prove \cref{thm:upper-general}.
\begin{proof}[Proof of \cref{thm:upper-general}]
  We apply \cref{prop:interval-bound} to each of the $z = C_{LZ77}(w)$ phrases of $w$.
  Summing the inequality of \cref{eq:interval_bound}  over $j=1,\dots,z$ gives
  \begin{align*}
    C_{LZ77}(w')
    &\le \sum_{j=1}^{z}\bigl(3 + \|g_{l_j}\| - \|g_{r_j}\| + 3n_j\bigr)&\\
    &= 3z + 3k + \sum_{j=1}^{z}\bigl(\|g_{l_j}\| - \|g_{r_j}\|\bigr)&\\
    &\le 3z + 3k +
    \sum_{j=1}^{z-1}\bigl(\|g_{l_j}\| - \|g_{l_{j+1}}\|\bigr) + \|g_{l_{z}}\| - \|g_{r_{z}}\| & (\text{by \cref{rk:monotone}, \cref{rk-monotone:4}})\\
    &\le
    3z + 3k + \corrindex{\bigl(\|g_1\| - \|g_{|w|+1}\| \bigr)} & (\text{by telescoping}) .\\
  \end{align*}
  By \cref{rk:monotone}~(4) and the definition, $\corrindex{k+1 \ge \|g_1\| \ge \|g_{|w|+1}\| \ge 1}$:
  the upper bound holds because $g_1$ is evaluated on the empty prefix $w'[1:0]$ and therefore increases
  its argument by at least one at every step, and the lower bound because $g_t(0)\le k+1$ for every $t$.
  Hence $\|g_1\|-\|g_{|w|+1}\| \le k$ and we have $C_{LZ77}(w') \le 3 C_{LZ77}(w) + 4k$, as claimed.
\end{proof}

\subsection{A general lower bound on \LZ sensitivity}

We now prove an almost-matching lower bound, showing that \cref{thm:upper-general} is tight up to lower-order terms.
\LowerGeneral
To this end, we construct an infinite family of words whose \LZ parse size asymptotically matches the lower bound stated in the proposition.
The construction combines two independent parts, $u$ and $v$, defined over disjoint alphabets. The first part is designed to contribute to an additive term of $k$ in $C(w')$, while the second part gives, again in $C(w')$, a multiplicative coefficient of $3$ in front of $C(w)$.
In what follows, fix $\eps > 0$. To lighten the notation we abbreviate $\CLZ{\cdot}$ by $C(\cdot)$ in
this subsection and in the appendix.

\subparagraph*{Construction of $u$.}

Let $m, k' \in \mathbb{N}$ be parameters to be chosen later.
We work over the alphabet $\Sigma = \{b, a_0, a_1, \dots, a_{k'}, \sharp_1, \sharp_2, \dots, \sharp_{k'}\}$, where all symbols are distinct.
Define:
\[
  u_i = \prod_{l=0}^{m} a_i b^l,
  \quad
  u = u_0 \cdot \prod_{l=1}^{k'} u_0 \sharp_l,
  \quad
  u' = u_0 \cdot \prod_{l=1}^{k'} u_l \sharp_l.
\]
The \LZ parses of $u$ and $u'$ are as follows:
\begin{align*}
  LZ77(u) &= a_0 \mid a_0 b \mid \cdots \mid a_0 b^m \mid u_0 \sharp_1 \mid u_0 \sharp_2 \mid \cdots \mid u_0 \sharp_{k'}, \\
  LZ77(u') &= a_0 \mid a_0 b \mid \cdots \mid a_0 b^m \mid
  a_1 \mid a_1 b \mid \cdots \mid a_1 b^m \mid \sharp_1 \mid
  \cdots \mid
  a_{k'} \mid a_{k'} b \mid \cdots \mid a_{k'} b^m \mid \sharp_{k'}.
\end{align*}
By replacing each $a_i$ in $u'$ with $a_0$, we obtain $u$, hence the Hamming distance between~$u$ and~$u'$ is
$k_u := (m+1)k'$; this is the number of edits used by this part of the construction.
Furthermore, following from the definition, we get that
$C(u) = m + 1 + k'$ and, counting for each $1\le l\le k'$ the $m+1$ phrases $a_l b^j$ together with the
phrase $\sharp_l$,
\[
  C(u') = (m+1)(k'+1) + k' = 3C(u) + \bigl(1 - f(m, k')\bigr)k_u,
  \qquad\text{where } f(m,k') = \frac{2}{k'} + \frac{2}{m+1} .
\]
For large enough $m$ and $k'$, $f(m,k') \le \eps$, and therefore we have
\begin{equation}
  C(u')\ge (3C(u) + k_u)\cdot (1 - \eps).\label{eq:compr-u}
\end{equation}

\subparagraph*{Construction of $v$.}
For this part, let the alphabet be $\Gamma = \{c, d, f, e_0, \dots, e_{p+1}\}$, where all symbols are distinct. Define:
\[
  v_{i,j} = c^i d d c^{i+1} e_j,
  \quad
  v = v_{p,p+1} \cdot \prod_{i=0}^{p} v_{i,i},
  \quad
  v' = c^p f f c^{p+1} e_{p+1} \cdot \prod_{i=0}^{p} v_{i,i}.
\]
The \LZ parses of $v$ and $v'$ are:
\begin{align*}
  LZ77(v) &= LZ77(v_{p,p+1}) \mid v_{0,0} \mid v_{1,1} \mid \cdots \mid v_{p,p}, \\
  LZ77(v') &= LZ77(c^p f f c^{p+1} e_{p+1}) \mid
  d \mid dc \mid e_0 \mid
  cd \mid dcc \mid e_1 \mid
  \cdots \mid
  c^p d \mid dc^{p+1} \mid e_p.
\end{align*}

The two headers have the same, explicit, parse: for $p\ge1$,
\[
  LZ77(c^p d d c^{p+1} e_{p+1}) = c \mid c^{p-1}d \mid dc \mid c^{p}e_{p+1}
  \text{ and }
  LZ77(c^p f f c^{p+1} e_{p+1}) = c \mid c^{p-1}f \mid fc \mid c^{p}e_{p+1} ,
\]
so that $s := C(v_{p,p+1}) = C(c^p f f c^{p+1} e_{p+1}) = 4$; note that the second phrase uses
self-reference.
The number of edits to get $v'$ from $v$ is $k_v = 2$, and we have
$C(v) = s + p + 1$, which grows to infinity with $p$.
Hence, the compression of $v'$ satisfies
\begin{align}
  C(v') &= s + 3(p + 1) = 3C(v) - 2s + k_v - 2&\notag\\
  &\ge (3C(v) + k_v)(1-\epsilon) &\text{for large enough } p .\label{eq:compr-v}
\end{align}

\subparagraph*{Combining both parts.}

Since the alphabets of $u$ and $v$ are disjoint, and the final blocks of $\LZ(u)$ and $\LZ(u')$ end exactly at the respective ends of $u$ and $u'$, the concatenations preserve parse independence, i.e. we have $C(uv) = C(u) + C(v)$ and $C(u'v') = C(u') + C(v')$.

Combining this observation with \cref{eq:compr-u} and \cref{eq:compr-v}, we get:
\begin{align*}
  C(u'v')
  &\ge (3C(u) + k_u)\cdot(1-\eps) + (3C(v) + k_v)\cdot(1-\eps) \\
  &= (3C(uv) + k)\cdot(1-\eps) ,
\end{align*}
where $k = k_u + k_v$ is the total number of edits.
The value $k_v$ is fixed, but we can choose $k_u$ arbitrarily, as long as $m$ and $k'$ are large enough; this concludes the proof of \cref{prop:lower-general}.

Let us stress why both parts are needed. Taken alone, the family $u$ forces
$k_u=(m+1)k'$ to be quadratic in $C(u)=m+1+k'$, so that the term $3C(u)$ is dominated by $k_u$ and the
multiplicative constant $3$ is not witnessed. Conversely, the family $v$ uses only $k_v=2$ edits, so it
witnesses the factor $3$ but says nothing about the additive term in $k$. Concatenating the two families
over disjoint alphabets makes both effects appear simultaneously, for an arbitrarily large total number
$k=k_u+2$ of edits.

\section{Fine-grained sensitivity upper and lower bounds}\label{sec:regimes}

\subsection{Upper bounds based on the compression ratio and number of edits}
\label{sec:upper-regime}

We refine our previous analysis by tightening the multiplicative coefficient of $\CLZ{w}$ in the bound for $\CLZ{w'}$.  While \cref{thm:upper-general} establishes a worst-case factor of $3$, we show that this factor can decrease to values near $2$ or even $1$, depending on the compressibility of $w$. This yields an interesting trichotomy in \textsc{LZ77}'s sensitivity, see \cref{cor:trichotomy}.

To formalize these regimes, we measure the proportion of original \LZ phrases in $w$ that are split into \corrindex{at least two, respectively at least three,} fragments by $k$ edits. Intuitively, if a large fraction of phrases of $w$ are split into two (respectively three) fragments in $w'$, then those phrases must heavily overlap common regions, otherwise a few edits could only impact a few phrases. Such overlap implies that $w$ is compressible (respectively highly compressible). \cref{thm:upper-regime} makes this intuition more precise.

\UpperRegime

\subsubsection{P-systems}
To prove \cref{thm:upper-regime}, we introduce two useful combinatorial objects: \emph{P2-system} and \emph{P3-system}, given in Definitions~\ref{def:p2} and~\ref{def:p3}. The key observation to motivate these definitions is that the set of original phrases containing at least two (respectively three) new phrase-starts in the edited word, \corrindex{once the at most $k$ phrases that carry an edited position have been discarded}, forms a \emph{P2-system} (respectively nearly a \emph{P3-system}). Hence, studying extremal combinatorial properties on these systems will translate to corresponding bounds on the compression of \LZ.

In what follows, each $(a,b)\in E\subset\mathbb{N}^2$ is viewed as a matching bracket pair, and each $i\in\sigma\subset\mathbb{N}$ as a marker.

\begin{definition}[P2-system]
  \label{def:p2}
  Let $E\subset\mathbb{N}^2$ and $\sigma\subset\mathbb{N}$.  The pair $(E,\sigma)$ is called a \emph{P2-system} if the following condition holds:

  Every bracket in $E$ contains at least one marker from $\sigma$: for each $(a,b)\in E$ there exists $i\in\sigma$ with $a\le i\corrindex{\le} b$.
\end{definition}

\begin{definition}[P3-system]
  \label{def:p3}
  Let $E\subset\mathbb{N}^2$ and $\sigma\subset\mathbb{N}$.  The pair $(E,\sigma)$ is a \emph{P3-system} if the following two conditions hold:
  \begin{enumerate}
    \item[\textup{(i)}] Every bracket in $E$ contains at least two markers from $\sigma$: for each $(a,b)\in E$ there exist $i,j\in\sigma$ with $a\le i<j\corrindex{\le} b$.
    \item[\textup{(ii)}] No two brackets share an endpoint. I.e., if $(a,b)\neq(c,d)$ are in $E$, then $a\neq c$ and $b\neq d$.
  \end{enumerate}
\end{definition}

\begin{figure}[ht]
  \centering
  \begin{tikzpicture}[x=0.48cm, y=0.48cm, >=Stealth]
    \begin{scope}[shift={(0,0)}]
      \draw[pipeblue, very thick] (0,0) to[out=80,in=100] (4,0);
      \draw[c4, very thick] (0,0) to[out=70,in=110] (8,0);
      \draw[c3, very thick] (5,0) to[out=80,in=100] (10,0);

      \foreach \x in {3, 6, 9}{
        \fill[editred] (\x,-0.15) circle (3pt);
      }
      \node[editred, font=\small, below] at (3,-0.3) {$i_1$};
      \node[editred, font=\small, below] at (6,-0.3) {$i_2$};
      \node[editred, font=\small, below] at (9,-0.3) {$i_3$};

      \node[font=\small\bfseries, anchor=north] at (5,-1.5) {(a) P2-system};
    \end{scope}

    \begin{scope}[shift={(14,0)}]
      \draw[pipeblue, very thick] (0,0) to[out=80,in=100] (5,0);
      \draw[c4, very thick] (1,0) to[out=80,in=100] (7,0);
      \draw[c3, very thick] (3,0) to[out=80,in=100] (10,0);

      \foreach \x in {2, 4, 6, 8}{
        \fill[editred] (\x,-0.15) circle (3pt);
      }
      \node[editred, font=\small, below] at (2,-0.3) {$i_1$};
      \node[editred, font=\small, below] at (4,-0.3) {$i_2$};
      \node[editred, font=\small, below] at (6,-0.3) {$i_3$};
      \node[editred, font=\small, below] at (8,-0.3) {$i_4$};

      \node[font=\small\bfseries, anchor=north] at (5,-1.5) {(b) P3-system};
    \end{scope}
  \end{tikzpicture}
  \caption{Illustration of P-systems. Brackets are arcs, markers are dots. (a)~In a P2-system, each bracket contains at least one marker. (b)~In a P3-system, each bracket contains at least two markers, and no two brackets share an endpoint. Markers are allowed to sit on the endpoints of a bracket, which is what happens in the application to \LZ parses.}
  \label{fig:p-systems}
\end{figure}

We now introduce the notion of \emph{extent} of a set $E \subseteq \mathbb{N}^2$.
This quantity will be used as a lower bound for $|w|$.
Later, by bounding the extent from below in terms of the number of edits $k$ and the number $p_2\cdot C(w)$ (respectively $p_3\cdot C(w)$) of \LZ phrases of $w$ that are split into \corrindex{at least two (respectively at least three)} parts in $w'$, we will obtain an upper bound on $C(w)$ that depends only on $k$, $|w|$ and $p_2$ (respectively $p_3$).

\begin{definition}[\emph{extent}]
  For $E\subset\mathbb{N}^2$, we define the \emph{extent} of $E$, written $\|E\|$, as follows:
  \[\|E\|=\sum_{(a, b)\in E}(b-a).\]
\end{definition}
In our proof of \cref{thm:upper-regime}, we will use the combinatorial properties of P-systems given by \cref{prop:p2} and \cref{prop:p3}. Both statements carry a linear correction term, which comes from brackets whose markers sit exactly on an endpoint; it is of lower order than the main term and, in the application, it is absorbed for free (see the proof of \cref{thm:upper-regime}).

\begin{proposition}
  \label{prop:p2}
  For $E\subset \nn^2$, $ \sigma\subset\nn$, if $(E, \sigma)$ is a {P2-system} then
  \[\|E\| \ge\frac{\card{E}^{3/2}}{8\sqrt {\card{\sigma}}}\corrindex{-\frac{\card{E}}{2}} .\]
\end{proposition}

\begin{proposition}
  \label{prop:p3}
  For $E\subset \nn^2$, $ \sigma\subset\nn$, if $(E, \sigma)$ is a {P3-system} then \[\|E\|\ge 2\frac{\card{E}^2}{\card{\sigma}}\corrindex{-\card{E}} .\]
\end{proposition}
We defer the proof of these propositions to the appendix, in \cref{app:p-systems}.

\subsubsection{Clean phrases}

Before deriving \cref{thm:upper-regime} from Propositions~\ref{prop:p2} and~\ref{prop:p3}, we isolate
the combinatorial properties of the parse of $w'$ that turn a family of \LZ phrases of $w$ into a
P-system.

Throughout this subsection we fix $w$ of length $n$, a set $\sigma=\{i_1<\cdots<i_k\}$ of edited
positions and $w'\in B(w,k)$ obtained from $w$ by substituting the positions of $\sigma$. For a phrase
$P_j$ of $\mathsf{LZ77}(w)$ with endpoints $l_j,r_j$ we write
$s_j = \card{\{i : l_j\le l(Q_i)\le r_j\}}$
for the number of phrase-starts of $\mathsf{LZ77}(w')$ it contains. We call $P_j$ \emph{clean} if
$[l_j,r_j]\cap\sigma=\emptyset$. Since the intervals $[l_j,r_j]$ partition $[1,n]$, at most $k$ phrases
of $w$ fail to be clean.

If $\lvert P_j\rvert\ge 2$, the greedy rule defining $\mathsf{LZ77}$ guarantees that the prefix
$w[l_j:r_j-1]$ occurs in $w$ at some position $l'_j\le l_j-1$. We call
\[
  (l'_j,r'_j),
  \qquad
  l'_j<l_j,
  \qquad
  r'_j = l'_j+(r_j-1-l_j)\le r_j-2,
  \qquad
  w[l'_j:r'_j]=w[l_j:r_j-1]
\]
a \emph{reference pair} of $P_j$ (any such pair will do). Note that $s_j\ge 2$ forces
$\lvert P_j\rvert \ge 2$, so every phrase we consider below has a reference pair.

The following lemma is the technical heart of the proof of \cref{thm:upper-regime}. Its point is that
\emph{the case where the interval $[l_j,r_j]$ itself carries an edited position has to be excluded}: for
clean phrases only, the number of new phrase-starts is controlled by the edits met by the reference
pair, and reference pairs behave rigidly.

\begin{lemma}\label{lem:clean}
  Let $P_j$ and $P_m$ be clean phrases of $w$ with reference pairs $(l'_j,r'_j)$ and $(l'_m,r'_m)$.
  \begin{enumerate}
    \item\label{it:clean-count} $s_j\le 1+\card{\sigma\cap[l'_j,r'_j]}$.
    \item\label{it:clean-inj} If $(l'_j,r'_j)=(l'_m,r'_m)$ and $j<m$, then $s_m\le 1$.
    \item\label{it:clean-left} If $l'_j=l'_m$, $r'_j<r'_m$, $s_j\ge3$ and $s_m\ge3$, then $j<m$ and
      $\sigma\cap(r'_j,r'_m]\ne\emptyset$.
    \item\label{it:clean-right} If $r'_j=r'_m$, $l'_j<l'_m$, $s_j\ge3$ and $s_m\ge3$, then $m<j$ and
      $\sigma\cap[l'_j,l'_m)\ne\emptyset$.
  \end{enumerate}
\end{lemma}

We defer the proof of \cref{lem:clean} to the appendix, in \cref{app:clean}.

\subsubsection{Compression upper bound from P-systems}

We will now show how \cref{thm:upper-regime}, the main result of this section, follows from Propositions~\ref{prop:p2} and~\ref{prop:p3}, together with \cref{lem:clean}.

\begin{proof}[Proof of \cref{thm:upper-regime}]
  We may assume $k\le n$, since $B(w,k)=B(w,n)$ as soon as $k\ge n$. We prove the two bounds separately.
  In what follows, let $C=\CLZ w$ and let
  $\sigma = \{i_\ell : 1\le \ell\le k\}$
  be the set of edited positions.

  \subparagraph*{The bound $p_2\cdot C \le 4 k^{1/3}n^{2/3}+k$.}
  Let $Z_2$ be the set of phrases $P_j$ of $w$ with $s_j\ge2$, so that $\card{Z_2}=p_2 C$, and let
  $Z_2^{\mathrm{cl}}\subseteq Z_2$ be the clean ones. As observed above, at most $k$ phrases of $w$ are
  not clean, so
  \begin{equation}\label{eq:drop-dirty-2}
    \card{Z_2^{\mathrm{cl}}} \ge p_2C-k .
  \end{equation}
  Every $P_j\in Z_2$ satisfies $\lvert P_j\rvert\ge2$ and hence has a reference pair $(l'_j,r'_j)$; let
  $
  E=\bigl\{(l'_j,r'_j) : P_j\in Z_2^{\mathrm{cl}}\bigr\} .
  $
  By \cref{lem:clean}\,(\ref{it:clean-inj}), two distinct phrases of $Z_2^{\mathrm{cl}}$ cannot share a
  reference pair: the later of the two would satisfy $s\le1$. Hence
  $\card E=\card{Z_2^{\mathrm{cl}}}\ge p_2C-k$.

  By \cref{lem:clean}\,(\ref{it:clean-count}), every $(a,b)\in E$ satisfies $\card{\sigma\cap[a,b]}\ge1$,
  that is, there is $i\in\sigma$ with $a\le i\le b$. Therefore $(E,\sigma)$ is a P2-system and
  \cref{prop:p2} applies:
  \[
    \|E\|  \ge \frac{\card{E}^{3/2}}{8\sqrt{k}}-\frac{\card E}{2} .
  \]
  On the other hand, the phrases of $w$ are pairwise disjoint and
  $r'_j-l'_j=r_j-l_j-1=\lvert P_j\rvert-2$, so
  \[
    \|E\|  = \sum_{P_j\in Z_2^{\mathrm{cl}}}\bigl(\lvert P_j\rvert-2\bigr) \le n-2\card E .
  \]
  Combining the two displays gives $n \ge \dfrac{\card{E}^{3/2}}{8\sqrt k}+\dfrac{3\card E}{2} \ge \dfrac{\card{E}^{3/2}}{8\sqrt k}$,
  hence $\card E\le 4k^{1/3}n^{2/3}$ and, by \eqref{eq:drop-dirty-2},
  \[
    p_2\cdot C \le 4k^{1/3}n^{2/3}+k ,
  \]
  as wanted.

  \subparagraph*{The bound $p_3\cdot C \le \frac{1}{\sqrt2}k^{3/2}\sqrt n+k$.}

  Let $Z_3$ be the set of phrases $P_j$ of $w$ with $s_j\ge3$, so that $\card{Z_3}=p_3C$, let
  $Z_3^{\mathrm{cl}}$ be the clean ones and let
  \[
    E_{\mathrm{ref}}=\bigl\{(l'_j,r'_j) : P_j\in Z_3^{\mathrm{cl}}\bigr\} ,
    \qquad
    \card{E_{\mathrm{ref}}}=\card{Z_3^{\mathrm{cl}}} \ge p_3C-k ,
  \]
  the equality again by \cref{lem:clean}\,(\ref{it:clean-inj}). By
  \cref{lem:clean}\,(\ref{it:clean-count}), every $(a,b)\in E_{\mathrm{ref}}$ satisfies
  $\card{\sigma\cap[a,b]}\ge2$. In particular, if $k\le1$ then $E_{\mathrm{ref}}=\emptyset$, so
  $p_3C\le k$ and the bound holds; we may therefore assume $k\ge2$.

  Contrary to the previous case, the set $E_{\text{ref}}$ is not automatically a P3-system, because some
  reference brackets may share an endpoint. We first bound how often this can happen.

  \begin{claim}\label{cl:degree}
    Every $a$ is the left endpoint of at most $k-1$ elements of $E_{\mathrm{ref}}$, and every $b$ is the
    right endpoint of at most $k-1$ elements of $E_{\mathrm{ref}}$.
  \end{claim}
  \begin{claimproof}
    Let $(a,b_1),\dots,(a,b_u)$ be the elements of $E_{\mathrm{ref}}$ with left endpoint $a$, ordered so
    that $b_1<\cdots<b_u$. By \cref{lem:clean}\,(\ref{it:clean-left}), for every $1\le v<u$ the interval
    $(b_v,b_{v+1}]$ contains a marker of $\sigma$. These $u-1$ intervals are pairwise disjoint and
    disjoint from $[a,b_1]$, which contains at least two markers. Hence $k\ge 2+(u-1)$, i.e.
    $u\le k-1$. The statement about right endpoints is symmetric, using
    \cref{lem:clean}\,(\ref{it:clean-right}): if $(a_1,b),\dots,(a_u,b)$ are ordered so that
    $a_1<\cdots<a_u$, then each $[a_v,a_{v+1})$ contains a marker and $[a_u,b]$ contains two.
  \end{claimproof}

  View $E_{\mathrm{ref}}$ as the edge set of a bipartite graph whose two sides are the left endpoints and
  the right endpoints occurring in $E_{\mathrm{ref}}$. By \cref{cl:degree} every vertex has degree at
  most $k-1$. By König's theorem the maximum matching number $\nu$ equals the minimum vertex cover
  number $\tau$; a vertex cover of size $\tau$ meets at most $\tau\cdot(k-1)$ edges, so
  $\card{E_{\mathrm{ref}}}\le\tau(k-1)=\nu(k-1)$. Let $E\subseteq E_{\mathrm{ref}}$ be a maximum
  matching; no two of its elements share an endpoint, so $(E,\sigma)$ is a P3-system, and
  \begin{equation}\label{eq:koenig}
    \card{E}\ge\frac{\card{E_{\mathrm{ref}}}}{k-1}\ge\frac{\card{E_{\mathrm{ref}}}}{k} .
  \end{equation}

  Applying \cref{prop:p3} to the P3-system $(E,\sigma)$ gives $\|E\|\ge 2\frac{\card E^2}{k}-\card E$.
  As in the previous case, \corrindex{$r'_j-l'_j=\lvert P_j\rvert-2$ for every
    $P_j\in Z_3^{\mathrm{cl}}$ and the phrases of $w$ are pairwise disjoint, so
    $\|E_{\mathrm{ref}}\|=\sum_{P_j\in Z_3^{\mathrm{cl}}}\bigl(\lvert P_j\rvert-2\bigr)
  \le n-2\card{E_{\mathrm{ref}}}$, i.e.}
  \begin{align*}
    n
    &\ge \|E_{\mathrm{ref}}\|+2\card{E_{\mathrm{ref}}}\\
    &\ge \|E\|+2\card{E}
    &(\text{since }E\subseteq E_{\mathrm{ref}})\\
    &\ge 2\frac{\card E^2}{k}+\card E
    &(\text{by \cref{prop:p3}})\\
    &\ge 2\frac{\card E^2}{k} .
  \end{align*}
  Hence $\card E\le\sqrt{nk/2}$ and, by \eqref{eq:koenig},
  $\card{E_{\mathrm{ref}}}\le k\card E\le \frac{k^{3/2}\sqrt n}{\sqrt2}$, so that
  \[
    p_3\cdot C\le\frac{k^{3/2}\sqrt{n}}{\sqrt 2}+k ,
  \]
  as wanted.
\end{proof}

\begin{remark}
  The additive term $k$ in \cref{thm:upper-regime} is the price of discarding the at most $k$ phrases of
  $w$ that carry an edited position. It does not affect the asymptotics in the regime of interest
  $n\gg k$.
\end{remark}

\subsection{Fine-grained sensitivity trichotomy}

We now prove the trichotomy.
Bounding only the fraction of split phrases ($p_2$ and $p_3$) is not enough, as a few outlier phrases could break into many pieces and inflate the overall ratio. To bound the total ratio $r$, we use an amortized argument: we apply \cref{prop:interval-bound} to bound the total number of extra fragments across all phrases.

\CorTrichotomy

\begin{proof}[Proof of \cref{cor:trichotomy}]
  Let $C = \CLZ{w}$ and $C' = \CLZ{w'}$. Let $s_j$ be the number of phrase-starts of $w'$ strictly contained within the interval of the $j$-th phrase of $w$, such that $C' = \sum_{j=1}^C s_j$.

  We decompose $s_j$ using indicator functions:
  $
  s_j \le 1 + \mathbf{1}_{\{s_j \ge 2\}} + \mathbf{1}_{\{s_j \ge 3\}} + \max(0, s_j - 3).
  $
  Summing over $j \in [C]$ yields $C' \le C + p_2(w,w')C + p_3(w,w')C + \sum_{j=1}^C \max(0, s_j - 3).$
  By \cref{prop:interval-bound}, $\max(0, s_j - 3) \le \|g_{l_j}\| - \|g_{r_j}\| + 3n_j$, where $n_j$ is the number of edits in the $j$-th phrase. Summing over all phrases yields a telescoping sum. Since $\sum n_j = k$, we obtain
  \[
    \sum_{j=1}^C \max(0, s_j - 3) \le \bigl(\|\corrindex{g_1}\| - \|\corrindex{g_{|w|+1}}\|\bigr) + 3k \le 4k.
  \]

  Dividing by $C$, the ratio $r = C'/C$ satisfies
  \begin{equation}\label{eq:ratio-bound}
    r \le 1 + p_2(w,w') + p_3(w,w') + \frac{4k}{C}.
  \end{equation}

  The three cases follow by evaluating \eqref{eq:ratio-bound} using \cref{thm:upper-regime}, which gives
  \[
    p_2(w,w')\le\min\left\{1,\ \frac{4k^{1/3}n^{2/3}+k}{C}\right\}
    \text{ and }
    p_3(w,w')\le\min\left\{1,\ \frac{1}{C}\left(\frac{k^{3/2}\sqrt n}{\sqrt 2}+k\right)\right\} .
  \]
  The standing hypothesis $C\ge 18k/\delta$ gives $k/C\le\delta/18$, hence $4k/C\le 2\delta/9$; this is where
  the assumption that $C$, and not merely $n$, is large relative to $k$ is used.

  If $C > \frac{12}{\delta}\max\{k^{1/3}n^{2/3},k^{3/2}\sqrt n\}$, then
  $\frac{4k^{1/3}n^{2/3}}{C}<\frac{\delta}{3}$ and $\frac{k^{3/2}\sqrt n}{\sqrt2\,C}<\frac{\delta}{12\sqrt2}$,
  so that $p_2(w,w')<\frac\delta3+\frac\delta{18}$ and $p_3(w,w')<\frac{\delta}{12\sqrt 2}+\frac\delta{18}$.
  Evaluated in \eqref{eq:ratio-bound}, this yields
  $r < 1 + \bigl(\frac13+\frac1{18}+\frac1{12\sqrt2}+\frac1{18}+\frac29\bigr)\delta < 1+\delta$.

  If $C > \frac{2k^{3/2}\sqrt{n}}{\delta}$, we apply the bound $p_2(w,w')\le 1$ and
  $p_3(w,w')<\frac{\delta}{2\sqrt2}+\frac\delta{18}$. Evaluated in \eqref{eq:ratio-bound}, this yields
  $r < 2 + \bigl(\frac1{2\sqrt2}+\frac1{18}+\frac29\bigr)\delta < 2 + \delta$.

  Otherwise, bounding $p_2(w,w') \le 1$ and $p_3(w,w') \le 1$ in \eqref{eq:ratio-bound} yields
  $r \le 3+\frac{4k}{C}\le 3+\frac{2\delta}{9}\le 3+\delta$, completing the proof.
\end{proof}

Finally, we give a lower bound that complements \cref{cor:trichotomy}.
\TrichotomyLower
The proof of \cref{prop:trichotomy-lower} is deferred to \cref{app:lb-two}.

\section{Conclusion and Future Work}

In this work, our results are stated for substitutions only. However, we expect them to extend to insertions and deletions with minor modifications to the proofs.

Several open directions remain. As already mentioned in the introduction, it remains open to close the gap in the second regime of the trichotomy (\cref{cor:trichotomy}), where the lower bound is off by a multiplicative factor of~$k$. While we suspect the upper bound to be close to optimal, establishing this formally would require additional work and complete the picture.

Finally, extending these questions to other compression schemes would provide a deeper theoretical understanding of standard data compression algorithms.

\bibliographystyle{plainurl}
\bibliography{biblio}

@inproceedings{BCR24,
  author    = {Bannai, Hideo and Charalampopoulos, Panagiotis and Radoszewski, Jakub},
  title     = {{Maintaining the Size of LZ77 on Semi-Dynamic Strings}},
  booktitle = {35th Annual Symposium on Combinatorial Pattern Matching (CPM 2024)},
  pages     = {3:1--3:20},
  series    = {Leibniz International Proceedings in Informatics (LIPIcs)},
  isbn      = {978-3-95977-326-3},
  issn      = {1868-8969},
  year      = {2024},
  volume    = {296},
  editor    = {Inenaga, Shunsuke and Puglisi, Simon J.},
  publisher = {Schloss Dagstuhl -- Leibniz-Zentrum f{\"u}r Informatik},
  address   = {Dagstuhl, Germany},
  url       = {https://drops.dagstuhl.de/entities/document/10.4230/LIPIcs.CPM.2024.3},
  urn       = {urn:nbn:de:0030-drops-201134},
  doi       = {10.4230/LIPIcs.CPM.2024.3}
}

@article{ZL77,
  title     = {A universal algorithm for sequential data compression},
  author    = {Ziv, Jacob and Lempel, Abraham},
  journal   = {IEEE Transactions on information theory},
  volume    = {23},
  number    = {3},
  pages     = {337--343},
  year      = {1977},
  publisher = {IEEE}
}

@inproceedings{Lopez05,
  title        = {Dimension is compression},
  author       = {L{\'o}pez-Vald{\'e}s, Mar{\'\i}a and Mayordomo, Elvira},
  booktitle    = {International Symposium on Mathematical Foundations of Computer Science},
  pages        = {676--685},
  year         = {2005},
  organization = {Springer}
}

@inproceedings{Lopez06,
  title        = {Lempel-ziv dimension for lempel-ziv compression},
  author       = {Lopez-Valdes, Maria},
  booktitle    = {International Symposium on Mathematical Foundations of Computer Science},
  pages        = {693--703},
  year         = {2006},
  organization = {Springer}
}

@article{Lutz03,
  title     = {Dimension in complexity classes},
  author    = {Lutz, Jack H},
  journal   = {SIAM Journal on Computing},
  volume    = {32},
  number    = {5},
  pages     = {1236--1259},
  year      = {2003},
  publisher = {SIAM}
}

@article{LM16,
  title   = {Computing absolutely normal numbers in nearly linear time},
  author  = {Lutz, Jack H and Mayordomo, Elvira},
  journal = {arXiv preprint arXiv:1611.05911},
  year    = {2016}
}

@article{ZL78,
  title     = {Compression of individual sequences via variable-rate coding},
  author    = {Ziv, Jacob and Lempel, Abraham},
  journal   = {IEEE transactions on Information Theory},
  volume    = {24},
  number    = {5},
  pages     = {530--536},
  year      = {1977},
  publisher = {IEEE}
}

@inproceedings{LathropS97,
  title        = {A universal upper bound on the performance of the Lempel-Ziv algorithm on maliciously-constructed data},
  author       = {Lathrop, James I and Strauss, Martin},
  booktitle    = {Proceedings. Compression and Complexity of SEQUENCES 1997 (Cat. No. 97TB100171)},
  pages        = {123--135},
  year         = {1997},
  organization = {IEEE}
}

@incollection{PierceS2000,
  author    = {Larry A. Pierce II and Paul C. Shields},
  title     = {Sequences Incompressible by SLZ (LZW), Yet Fully Compressible by ULZ},
  booktitle = {Numbers, Information and Complexity},
  pages     = {385--390},
  address   = {Boston},
  year      = {2000}
}

@inproceedings{LagardeP18,
  author    = {Guillaume Lagarde and
               Sylvain Perifel},
  editor    = {Artur Czumaj},
  title     = {Lempel-Ziv: a "one-bit catastrophe" but not a tragedy},
  booktitle = {Proceedings of the Twenty-Ninth Annual {ACM-SIAM} Symposium on Discrete
               Algorithms, {SODA} 2018, New Orleans, LA, USA, January 7-10, 2018},
  pages     = {1478--1495},
  publisher = {{SIAM}},
  year      = {2018},
  url       = {https://doi.org/10.1137/1.9781611975031.97},
  doi       = {10.1137/1.9781611975031.97},
  bibsource = {dblp computer science bibliography, https://dblp.org}
}

@article{nakashima24,
  title   = {Edit and Alphabet-Ordering Sensitivity of Lex-Parse},
  author  = {Nakashima, Yuto and K{\"o}ppl, Dominik and Funakoshi, Mitsuru and Inenaga, Shunsuke and Bannai, Hideo},
  journal = {arXiv preprint arXiv:2402.19223},
  year    = {2024}
}

@article{AkagiFI23,
  author    = {Tooru Akagi and
               Mitsuru Funakoshi and
               Shunsuke Inenaga},
  title     = {Sensitivity of string compressors and repetitiveness measures},
  journal   = {Inf. Comput.},
  volume    = {291},
  pages     = {104999},
  year      = {2023},
  url       = {https://doi.org/10.1016/j.ic.2022.104999},
  doi       = {10.1016/J.IC.2022.104999},
  bibsource = {dblp computer science bibliography, https://dblp.org}
}

@article{Giuliani2025,
  author    = {Sara Giuliani and
               Shunsuke Inenaga and
               Zsuzsanna Lipt{\'{a}}k and
               Giuseppe Romana and
               Marinella Sciortino and
               Cristian Urbina},
  title     = {Bit Catastrophes for the Burrows-Wheeler Transform},
  journal   = {Theory Comput. Syst.},
  volume    = {69},
  number    = {2},
  pages     = {19},
  year      = {2025},
  url       = {https://doi.org/10.1007/s00224-024-10212-9},
  doi       = {10.1007/S00224-024-10212-9},
  bibsource = {dblp computer science bibliography, https://dblp.org}
}

@article{Rytter03,
  author    = {Wojciech Rytter},
  title     = {Application of Lempel-Ziv factorization to the approximation of grammar-based
               compression},
  journal   = {Theor. Comput. Sci.},
  volume    = {302},
  number    = {1-3},
  pages     = {211--222},
  year      = {2003},
  url       = {https://doi.org/10.1016/S0304-3975(02)00777-6},
  doi       = {10.1016/S0304-3975(02)00777-6},
  bibsource = {dblp computer science bibliography, https://dblp.org}
}

@inproceedings{MatiasRS99,
  title     = {On the optimality of parsing in dynamic dictionary based data compression},
  author    = {Matias, Yossi and {\c{S}}ahinalp, S{\"u}leyman Cenk},
  booktitle = {Proceedings of the tenth annual ACM-SIAM symposium on Discrete algorithms},
  pages     = {943--944},
  year      = {1999}
}

@article{MatiasWAE,
  title   = {Implementation and experimental evaluation of flexible parsing for dynamic dictionary based data compression},
  author  = {Matias, Yossi and Rajpoot, Nasir Mahmood and Sahinalp, Suleyman Cenk},
  journal = {Proceedings WAE’98},
  year    = {2008}
}

@article{AronicaLMMM18,
  title     = {On optimal parsing for LZ78-like compressors},
  author    = {Aronica, Salvatore and Langiu, Alessio and Marzi, Francesca and Mazzola, Salvatore and Mignosi, Filippo},
  journal   = {Theoretical Computer Science},
  volume    = {710},
  pages     = {19--28},
  year      = {2018},
  publisher = {Elsevier}
}

@phdthesis{OnOptimalParsing,
  title  = {Optimal Parsing for dictionary text compression},
  author = {Langiu, Alessio},
  year   = {2012},
  school = {Universit{\'e} Paris-Est; Universit{\`a} degli studi (Palerme, Italie)}
}

\appendix

\section{Properties of the jump-depth}\label{app:jump-depth}
We now turn to proving \cref{prop:interval-bound}.
We introduce the following lemma:
\begin{lemma}\label{lem:jump-gap}
  Let $\corrindex{1 \le j < j' \le |w|+1}$, and let $\ell \ge 0$, $m \ge s \ge 2$ be integers satisfying
  \[
    \corrindex{\ell + m \le k} ,
    \qquad
    g_j^s(\ell) \le \ell + m
    \text{ and }
    g_{j'}(\ell+1) \ge \ell + m.
  \]
  Then, we have $\|g_j\| - \|g_{j'}\| \ge s - 2$.
\end{lemma}

\begin{figure}[htbp]\centering
  \begin{tikzpicture}[scale=1]

    \draw[->] (-0.5,0) -- (9,0);

    \draw[dashed, thick, gray] (0,0) -- (0,3.2);
    \draw[dashed, thick, gray] (8,0) -- (8,3.2);

    \foreach \x/\label in {0/$l$, 1/$l+1$, 8/$l+m$} {
      \draw (\x,0.1) -- (\x,-0.1);
      \node[below] at (\x,0) {\label};
    }

    \def\length{1.7}

    \foreach \i in {0,1,2} {
      \pgfmathsetmacro{\startx}{\i*\length}
      \pgfmathsetmacro{\endx}{(\i+1)*\length}
      \draw[thick, green!60!black, |-] (\startx,1) -- (\endx,1);
      \draw[thick, green!60!black, ->] (\endx,1) -- ({\endx + 0.01},1);

      \draw[very thick, green] (\startx,1) -- (\startx,1.5);
    }
    \draw[very thick, green] (\length,1) -- (\length,2.5);
    \node[black, font=\tiny] at (-0.3,1.3) {$g_j^{t}(0)$};
    \node[black, font=\tiny] at (\length-0.45,1.3) {$g_j^{t+1}(0)$};
    \node[black, font=\tiny] at (2*\length-0.45,1.3) {$g_j^{t+2}(0)$};
    \pgfmathsetmacro{\laststart}{8-\length}
    \draw[thick, green!60!black, |-] (\laststart,1) -- (8,1);
    \draw[thick, green!60!black, ->] (8,1) -- (8.01,1);
    \draw[very thick, green] (8,1) -- (8,1.5);
    \node[black, font=\tiny] at (7.55,1.3) {$g_j^{t+s}(0)$};

    \pgfmathsetmacro{\dotsx}{(2*\length + \laststart)/2}
    \node[green!60!black] at (\dotsx+0.90,1) {\Huge $\cdots$};

    \draw[decorate, decoration={brace, mirror}, thick] (0,0.75) -- (8,0.75)node[midway, below=8pt, font=\small\itshape] {$s$};

    \draw[thick, red!70!black, |-] (1,2) -- (8,2);
    \draw[thick, red!70!black, ->] (8,2) -- (8.01,2);

    \draw[very thick, red] (8,2) -- (8,2.5);

    \node[black, font=\tiny] at (2.15,2.25) {$g_{j'}^{t+1}(0)$};
    \node[black, font=\tiny] at (8.45,2.25) {$g_{j'}^{t+2}(0)$};

    \node[anchor=west, text width=1.2cm, align=right] at (-2.5,1) {$g_j$ :};
    \node[anchor=west, text width=1.2cm, align=right] at (-2.5,2) {$g_{j'}$ :};

  \end{tikzpicture}
  \caption{Positions and relations of values used in Lemma~\ref{lem:jump-gap}}
  \label{fig:jump-gap}
\end{figure}

\begin{proof}
  \Cref{fig:jump-gap} illustrates the arguments of this proof.
  Define
  \begin{align}
    t & = \max\{u\ge0 : g_j^u(0)\corrindex{\le}\ell\}, \label{eq:def-t}\\
    o & = \min\{u\ge0 : g_j^{t+s+u}(0)\ge \corrindex{k+1}\}. \label{eq:def-o}
  \end{align}
  Both are well defined: $g_j^0(0)=0\le\ell$ and $\ell\le k-2$, so $t$ is finite, and
  by maximality of $t$ we have
  \begin{equation}\label{eq:t-plus-one}
    g_j^{t}(0)\le \ell
    \quad\text{and}\quad
    g_j^{t+1}(0)\ge \ell+1 .
  \end{equation}
  We first check that $\|g_j\|=t+s+o$. By \cref{cl:upper-bound} below, $g_j^{t+s}(0)\le \ell+m\le k<k+1$,
  so $o\ge1$; minimality of $o$ then gives $g_j^{t+s+o-1}(0)<k+1$, and since $u\mapsto g_j^u(0)$ is
  non-decreasing, $t+s+o$ is indeed the least $u$ with $g_j^{u}(0)\ge k+1$.
  It therefore suffices to show that $\|g_{j'}\| \le t+2+o$,
  i.e. that~$g_{j'}^{t+2+o}(0) \ge k+1$.

  We prove this using the following two claims.
  \begin{claim}\label{cl:lower-bound}
    $g_{j'}^{t+2}(0)\ge \ell+m$.
  \end{claim}
  \begin{claimproof}
    Using \cref{rk:monotone}, items~(3) and~(5):
    \[
      \begin{aligned}
        g_{j'}^{t+2}(0)
        &=g_{j'}\bigl(g_{j'}^{t+1}(0)\bigr)
        \ge
        g_{j'}\bigl(g_j^{t+1}(0)\bigr)
        &&(\text{by (3) and (5)})\\
        &\ge
        g_{j'}(\ell+1)
        &&(\text{by (3) and \cref{eq:t-plus-one}})\\
        &\ge
        \ell+m
        &&(\text{by hypothesis}).
      \end{aligned}
    \]
    Note that the second inequality uses $g_j^{t+1}(0)\ge \ell+1$, which is what the
    definition~\eqref{eq:def-t} of $t$ provides; the weaker bound $g_j^{t+1}(0)\ge\ell$ obtained from
    the strict variant $t=\max\{u : g_j^u(0)<\ell\}$ would not suffice here.
  \end{claimproof}

  \begin{claim}\label{cl:upper-bound}
    $g_j^{t+s}(0)\;\le\;\ell+m$.
  \end{claim}
  \begin{claimproof}
    By definition of $t$ and \cref{rk:monotone}~(3),
    \begin{align*}
      g_j^{t+s}(0)
      &=g_j^s\bigl(g_j^{t}(0)\bigr)
      \le
      g_j^s(\ell)
      &&(\text{by (3) and \cref{eq:t-plus-one}})\\
      &\le
      \ell+m
      &&(\text{by hypothesis} ).
    \end{align*}
  \end{claimproof}

  Finally, combining Claims~\ref{cl:lower-bound} and \ref{cl:upper-bound}:
  \begin{align*}
    g_{j'}^{t+2+o}(0)
    &=g_{j'}^o\bigl(g_{j'}^{t+2}(0)\bigr)
    \ge
    g_{j'}^o(\ell+m)
    &&(\text{by Claim~\ref{cl:lower-bound} and (3)})\\
    &\ge
    g_j^o(\ell+m)
    &&(\text{by (5)})\\
    &\ge g_j^{t+s+o}(0)
    &&(\text{by Claim~\ref{cl:upper-bound} and (3)})\\
    &\ge
    k+1
    &&(\text{by definition \eqref{eq:def-o}}),
  \end{align*}
  as required.
  \qedhere
\end{proof}

We are now ready to prove \cref{prop:interval-bound}.
\begin{proof}[Proof of \cref{prop:interval-bound}]
  Let $s = \card{\{i: l_j\le l(Q_i)\le r_j\}}$
  denote the number of new phrase-starts in the interval $[l_j,r_j]$. We consider two cases based on $n_j$, the number of edits in this interval.

  \subparagraph*{Case 1:} $n_j=0$.
  \begin{itemize}

    \item If $s\le3$, then $s \le 3 \le 3 + \|g_{l_j}\| - \|g_{r_j}\|$ since $\corrindex{\|g_{r_j}\| \le \|g_{l_j}\|}$ by \cref{rk:monotone}~(4).

    \item If $s>3$: choose any $x<y \le r_j -2$ so that $w[x:y]=w'[l_j:r_j-1]$. Such $(x,y)$ exists because the \corrindex{proper prefix} $w[l_j:r_j-1]$ of the phrase $P_j$ must appear previously in $w$.
      We suppose that $w'[x:y]$ contains $m$ edited positions, i.e., $w'$ is of the form:
      \[w'[x:y] = u_0\#_{\ell+1}u_1\#_{\ell+2}u_2\cdots u_{m-1}\#_{\ell+m}u_m\]
      where $\#_{\ell+1}, \cdots, \#_{\ell+m}$ are the $m$ edited positions, of indices $i_{\ell+1}, i_{\ell+2}, \cdots, i_{\ell+m}$.
      \corrindex{Since $i_{\ell+m}\le y\le n$, we have $\ell+m\le k$, so the index $\ell+m+1$ used below is a legitimate
        state of the jump-map: either $i_{\ell+m+1}$ is a genuine substitution position, or $\ell+m=k$ and
        $i_{\ell+m+1}=i_{k+1}=n+1$ is the right sentinel of \cref{def:jump-function}. This is exactly the
        boundary condition required by \cref{lem:jump-gap}, and it is the reason why the absorbing state of
      the jump-map is $k+1$ rather than $k$.}

      We now show that
      \begin{itemize}
        \item $g_{r_j}(\ell+1)\ge \ell+m$.
        \item $g_{l_j}^{s-1}(\ell)<\ell+m+1$.
      \end{itemize}

      For the first point, since
      $$
      w[i_{\ell+1}+1: i_{\ell+m}-1] \factor w[x:y] = w'[l_j:r_j-1],
      $$
      and $w'[l_j:r_j-1]$ is a factor of $w'[\corrindex{1}:r_j-1]$, we may apply the definition of $g_{r_j}$ directly to conclude $g_{r_j}(\ell+1)\ge \ell+m$.

      Before proving the second point, we record two observations.
      \begin{itemize}
        \item[$(\alpha)$] \emph{All the substitutions met by the reference occur before $l_j$}, i.e.
          $i_{\ell+m}<l_j$. Indeed $i_{\ell+m}\in[x,y]$ and $[x,y]\subseteq[1,r_j-2]$; since we are in
          Case~1, $[l_j,r_j]$ contains no substitution, so $i_{\ell+m}\notin[l_j,r_j]$, and therefore
          $i_{\ell+m}<l_j$.
        \item[$(\beta)$] \emph{Consequently, for every $\ell\le q<\ell+m$ the second branch of
          \cref{def:jump-function} is the one that applies to $g_{l_j}(q)$.} Indeed
          $w[i_{q}+1:i_{q+1}-1]$ contains no substituted position, so it equals
          $w'[i_{q}+1:i_{q+1}-1]$, and $i_{q+1}\le i_{\ell+m}<l_j$ by $(\alpha)$, so it is a factor of
          $w'[1:l_j-1]$. In particular, for such a $q$, the jump from $q$ to $g_{l_j}(q)$ really does
          certify that $w[i_{q}+1:i_{g_{l_j}(q)}-1]$ is a factor of $w'[1:l_j-1]$ --- this is what the
          base case $g_{l_j}(q)=q+1$ of the definition does \emph{not} guarantee in general.
      \end{itemize}

      For the second point, suppose for the sake of contradiction that
      $$
      g_{l_j}^{s-1}(\ell) \ge \ell+m+1.
      $$
      Let $k_0 = \ell$, and for $1 \le v \le s-1$, let $k_v = g_{l_j}(k_{v-1})$ be the successive jump
      indices, so that $k_{s-1}\ge\ell+m+1$, and let
      \corrindex{$v^\star = \min\{v\ge1 : k_v\ge \ell+m+1\}\le s-1$}.
      See Figure~\ref{fig:case1_decomposition} for an
      illustration. We claim that $w'[l_j:r_j-1]$ can be written as
      \begin{equation}\label{eq:blocks}
        w'[l_j:r_j-1] = f_1c_1f_2c_2\cdots c_{v^\star-1}f_{v^\star} ,
      \end{equation}
      where each $c_v$ is a single character and each $f_v$ (possibly empty) has an occurrence in $w'$
      \corrindex{starting at a position $<l_j$}.

      \corrindex{Indeed, recall that $w'[l_j:r_j-1]=w[x:y]$ with $i_\ell<x$ and $y<i_{\ell+m+1}$. By
        minimality of $v^\star$ we have $k_v\le\ell+m$ for every $v<v^\star$, so
        $i_{k_1},\dots,i_{k_{v^\star-1}}$ all belong to $[x,y]$: cut $w[x:y]$ at these positions, letting
        $c_v=w[i_{k_v}]$ and $f_v$ be the stretch of $w$ strictly between two consecutive cuts, truncated
        to $[x,y]$. For $v<v^\star$ we have $k_{v-1}<k_v\le\ell+m$, so $(\beta)$ applies and
        $w[i_{k_{v-1}}+1:i_{k_v}-1]$ is a factor of $w'[1:l_j-1]$; so is $f_v$, being one of its factors.
        For the last block, $f_{v^\star}$ is a prefix of $w[i_{k_{v^\star-1}}+1:y]$, and we distinguish
        two cases. If $k_{v^\star-1}<\ell+m$, then $(\beta)$ applies again and
        $w[i_{k_{v^\star-1}}+1:i_{k_{v^\star}}-1]$ is a factor of $w'[1:l_j-1]$; since
        $i_{k_{v^\star}}\ge i_{\ell+m+1}>y$, the block $f_{v^\star}$ is a prefix of it. Otherwise
        $k_{v^\star-1}=\ell+m$, and the positions $i_{\ell+m}+1,\dots,y$ carry no substitution, so
        $w[i_{\ell+m}+1:y]=w'[i_{\ell+m}+1:y]$, an occurrence starting at $i_{\ell+m}+1\le l_j-1<l_j$ by
        $(\alpha)$. Note that in this last case the occurrence may extend beyond $l_j$, which is harmless:
        \LZ with self-reference only requires the source of a match to \emph{start} before the matched
      position.}

      \corrindex{We now conclude with the following counting principle, applied to
      \eqref{eq:blocks}.}

      \begin{claim}\label{cl:blocks}
        \corrindex{Let $a\le b<|w'|$ and suppose $w'[a:b]=f_1c_1\cdots c_{N-1}f_{N}$, where each $c_v$
          is a single character and each $f_v$ occurs in $w'$ at some position $<a$. Then $w'$ has at most
        $N$ phrase-starts in $[a,b+1]$.}
      \end{claim}
      \begin{claimproof}
        \corrindex{Group the decomposition into blocks $B_v=f_vc_v$ for $v<N$ and $B_{N}=f_{N}$, which
          partition $[a,b]$. Let $p$ be a phrase-start with $p\le b$, lying in $B_v$. The suffix of $f_v$
          starting at $p$ (possibly empty) occurs in $w'$ at a position $<a\le p$, so the greedy match
          from $p$ reaches at least the end of $f_v$; adding the one extra character of a \LZ phrase, the
          phrase starting at $p$ ends at position at least the end of $B_v$. Hence the next phrase-start
          lies in a block of strictly larger index, and if $v=N$ it is $>b+1$. Therefore the blocks
          containing phrase-starts have strictly increasing indices, and a phrase-start at $b+1$ can only
          follow one in a block of index $\le N-1$. In both cases there are at most $N$ phrase-starts in
        $[a,b+1]$.}
      \end{claimproof}

      \corrindex{Applying \cref{cl:blocks} to \eqref{eq:blocks} with $a=l_j$ and $b=r_j-1$ gives at most
      $v^\star\le s-1$ phrase-starts in $[l_j,r_j]$, contradicting the fact that there are $s$ of them.}

      \corrindex{Note finally that $g_{l_j}^{s-1}(\ell)\le\ell+m$ also yields $m\ge s-1$: as long as the
        iterates stay below the absorbing state they increase by at least one at each step
        (\cref{rk:monotone}~(1)), so $\ell+s-1\le g_{l_j}^{s-1}(\ell)\le \ell+m$. This is the last
      hypothesis needed to apply \cref{lem:jump-gap} with $s-1$ jumps.}

      Applying Lemma~\ref{lem:jump-gap} with the pair $(l_j,r_j)$, the value $\ell$, the bound $m$ and $s-1$ jumps --- its hypotheses $\ell+m\le k$, $g_{l_j}^{s-1}(\ell)\le \ell+m$ and $g_{r_j}(\ell+1)\ge \ell+m$ being the three facts just established --- we get:
      \[
        \|g_{l_j}\|-\|g_{r_j}\|\ge(s-1)-2
        =s-3,
      \]
      hence
      \[
        s\le3+\bigl(\|g_{l_j}\|-\|g_{r_j}\|\bigr).
      \]
  \end{itemize}
  \begin{figure}[ht]\centering
    \begin{tikzpicture}[scale=1]

      \draw[thick, ->] (-1, 0) -- (11.5, 0);

      \draw[dashed, thick, blue] (1.2, -0.6) -- (1.2, 2) node[below, pos=0] {$x$};
      \draw[dashed, thick, blue] (9.2, -0.6) -- (9.2, 2) node[below, pos=0] {$y$};
      \draw[decorate, decoration={brace, amplitude=8pt}, thick, blue] (1.2, 2.1) -- (9.2, 2.1) node[midway, above=10pt] {Match of the prefix $w'[l_j:r_j-1]$};

      \foreach \pos/\lab in {0/i_{\ell}, 2/i_{\ell+1}, 4/i_{\ell+2}, 6/i_{\ell+3}, 8/i_{\ell+m}, 10/i_{\ell+m+1}} {
        \draw[thick, editred] (\pos, -0.15) -- (\pos, 0.15);
        \node[editred, below=12pt] at (\pos, 0) {$\lab$};
        \node[editred] at (\pos, 0) {$\times$};
      }
      \node[editred] at (7, -0.42) {$\dots$};

      \draw[->, thick, c3, bend left=45] (0, 0.2) to node[above, font=\small] {$f_1$} (4, 0.2);
      \draw[->, thick, c3, bend left=45] (4, 0.2) to node[above, font=\small] {$f_2$} (6, 0.2);
      \draw[->, thick, c3, bend left=45] (8, 0.2) to node[above, font=\small] {$f_{s-1}$} (10, 0.2);
      \node[c3] at (7, 0.95) {$\dots$};

      \node[c3, font=\small, below=26pt] at (4,0) {$=i_{k_1}$};
      \node[c3, font=\small, below=26pt] at (6,0) {$=i_{k_2}$};
      \node[c3, font=\small, below=26pt] at (10,0) {$=i_{k_{s-1}}$};

    \end{tikzpicture}
    \caption{Construction of the decomposition in \cref{prop:interval-bound} (Case 1 for $s>3$). The target substring $w[x:y]$ corresponds to the phrase prefix $w'[l_j:r_j-1]$. The sequence of $s-1$ jumps of $g_{l_j}$ (green arcs) spans $[x,y]$. \corrindex{Each jump endpoint $i_{k_v}$ is itself a substitution position, and the character $w[i_{k_v}]$ sitting there is the separator $c_v$ of \eqref{eq:blocks}: a block of the decomposition is one factor followed by one character, which is exactly the shape of a \LZ phrase.}}
    \label{fig:case1_decomposition}
  \end{figure}
  \subparagraph*{Case 2:} $n_j>0$.
  The goal is to reduce this case to the first one by cutting $w'[l_j:r_j]$ according to the edited positions.

  Let
  \[
    w'[l_j:r_j] = w_0 \#_{\ell} w_{1} \#_{\ell+1} \cdots w_{n_j-1}\#_{\ell+n_j-1}w_{n_j}
  \]
  with $w_0, w_1, \cdots, w_{n_j}\in (\Sigma\setminus\{\#_i, 1\le i\le k\})^*$.

  For $0\le o<n_j$, let
  \[
    \begin{array}{rclclclcl}
      l_{j,o}   &=& l_j + \displaystyle\sum_{m=0}^{o-1} |w_m| + o
      && \text{and} &&
      r_{j,o}   &=& l_j + \displaystyle\sum_{m=0}^{o} |w_m| + o, \\
      l_{j,n_j} &=& r_j - |w_{n_j}| + 1
      && \text{and} &&
      r_{j,n_j} &=& r_j.
    \end{array}
  \]

  $[l_{j, o}, r_{j, o}]$ is the interval associated with $w_o \#_{\ell+o}$ and $[l_{j, n_j}, r_{j, n_j}]$ with $w_{n_j}$.

  \corrindex{Observe that Case~1 uses exactly two properties of the interval $[l_j,r_j]$: that
    $w'[l_j: r_j-1]$ occurs in $w[1: r_j-2]$ --- we never need $w[l_j: r_j]$ to be a \LZ factor of $w$
    --- and that $[l_j,r_j-1]$ carries no substituted position, the latter being what makes
  $(\alpha)$ and $(\beta)$ available. Both hold for every sub-interval $[l_{j,o},r_{j,o}]$.}

  By definition of a factor in $\mathsf{LZ77}(w)$, $w[l_j: r_j-1]$ has at least one occurrence in $w[\corrindex{1}: r_j-2]$ so each $w_o$ occurs in $w[\corrindex{1}: r_{j,o}-2]$. \corrindex{Moreover, by construction the only substituted position of $[l_{j,o},r_{j,o}]$ is its last one, $r_{j,o}$, which plays there the role that the trailing character $r_j$ of the phrase plays in Case~1; in particular $(\alpha)$ still holds, since a substitution met by the reference lies in $[1,r_{j,o}-2]$ and, not being in $[l_{j,o},r_{j,o}-1]$, must lie before $l_{j,o}$.} Hence, it is possible to apply the reasoning used in Case 1 on all the intervals $[l_{j, o}: r_{j, o}]$ with $0\le o \le n_j$.

  Therefore,
  \begin{align*}
    s
    &\le
    \sum_{o=0}^{n_j} \bigl(3+\|g_{l_{j,o}}\|-\|g_{r_j,o}\|\bigr)&\\
    &\le
    \sum_{o=0}^{n_j-1}  \bigl(3+\|g_{l_{j,o}}\|-\|g_{l_j,o+1}\|\bigr) + 3 + \|g_{l_{j,n_j}}\|-\|g_{r_j,n_j}\| & (\text{by \cref{rk:monotone}~(4)})\\
    &\le
    3+\|g_{l_{j}}\|-\|g_{r_j}\|+3n_j & (\text{by telescoping})\\
  \end{align*}
  as wanted.
\end{proof}

\section{Clean phrases}\label{app:clean}

We now prove \cref{lem:clean}.

\begin{proof}
  For $p\in[l_j,r_j-1]$ set $\mathsf{src}_j(p)=l'_j+(p-l_j)$, so that $\mathsf{src}_j(p)<p$. We shall
  repeatedly use the following observation: for $l_j\le p\le p'\le r_j-1$ with
  $\sigma\cap[\mathsf{src}_j(p),\mathsf{src}_j(p')]=\emptyset$,
  \begin{equation}\label{eq:clean-copy}
    w'[p:p']
    = w[p:p']
    = w[\mathsf{src}_j(p):\mathsf{src}_j(p')]
    = w'[\mathsf{src}_j(p):\mathsf{src}_j(p')] ,
  \end{equation}
  where the first equality holds because $P_j$ is clean, the second because
  $w[l'_j:r'_j]=w[l_j:r_j-1]$, and the third because the source range carries no edit. In particular
  $w'[p:p']$ has an occurrence in $w'$ starting strictly before $p$, so the phrase of
  $\mathsf{LZ77}(w')$ starting at $p$ covers $[p,p']$ and one further character, hence ends at
  position at least $p'+1$.

  \smallskip\noindent\emph{Item~\ref{it:clean-count}.}
  Let $A=\{p\in[l_j,r_j-1] : \mathsf{src}_j(p)\in\sigma\}$, so that $\card A=\card{\sigma\cap[l'_j,r'_j]}=:t$.
  Let $p_1<\cdots<p_{s_j}$ be the phrase-starts of $w'$ in $[l_j,r_j]$ and fix $u<s_j$; then
  $p_u\le p_{s_j}-1\le r_j-1$. Let $p^\star_u=\min\{q\in A : q\ge p_u\}$. If $p^\star_u$ did not exist,
  \eqref{eq:clean-copy} applied to $(p_u,r_j-1)$ would show that the phrase starting at $p_u$ ends at
  position at least $r_j$, contradicting $p_{u+1}\le r_j$. So $p^\star_u$ exists, and
  \eqref{eq:clean-copy} applied to $(p_u,p^\star_u-1)$ shows that the phrase starting at $p_u$ ends at
  position at least $p^\star_u$ --- when $p^\star_u=p_u$ this is vacuously true, a phrase having length at
  least one --- whence $p_{u+1}>p^\star_u$. Consequently
  $p^\star_1<p^\star_2<\cdots<p^\star_{s_j-1}$ are $s_j-1$ distinct elements of $A$, and $s_j-1\le t$.

  \smallskip\noindent\emph{Item~\ref{it:clean-inj}.}
  Equal reference pairs give $w[l_j:r_j-1]=w[l'_j:r'_j]=w[l'_m:r'_m]=w[l_m:r_m-1]$. Both phrases being
  clean, $w'[l_m:r_m-1]=w[l_m:r_m-1]=w[l_j:r_j-1]=w'[l_j:r_j-1]$. As $j<m$ we have $r_j<l_m$, so for
  every $p\in[l_m,r_m-1]$ the string $w'[p:r_m-1]$ occurs in $w'$ at position
  $l_j+(p-l_m)\le r_j-1<p$. The phrase starting at any such $p$ therefore ends at position at least
  $r_m$, and $s_m\le1$.

  \smallskip\noindent\emph{Item~\ref{it:clean-left}.}
  Assume first $m<j$. Then $w'[l_j:r_j-1]=w[l'_j:r'_j]$ is a prefix of $w[l'_m:r'_m]=w'[l_m:r_m-1]$,
  which occurs at $l_m<l_j$; arguing as in item~(\ref{it:clean-inj}) gives $s_j\le1$, a contradiction.
  Hence $j<m$, and $r_j<l_m$.

  Suppose now that $\sigma\cap(r'_j,r'_m]=\emptyset$ and set $d=r'_j-l'_m=r'_j-l'_j\ge0$, so that
  $l_m+d<r_m-1$. \corrindex{We bound separately the phrase-starts lying in $[l_m,l_m+d+1]$ and those
    lying in $[l_m+d+2,r_m]$; note that $l_m$ itself need not be a phrase-start of $w'$.
    Let $p\in[l_m,l_m+d]$ be a phrase-start. Since $P_m$ and $P_j$ are clean and $l'_m=l'_j$,
    \[
      w'[p:l_m+d]=w[\mathsf{src}_m(p):r'_j]
      =w'\bigl[l_j+(\mathsf{src}_m(p)-l'_j):r_j-1\bigr] ,
    \]
    an occurrence starting at position at most $r_j-1<l_m\le p$; the phrase starting at $p$ therefore
    ends at position at least $l_m+d+1$, and the next phrase-start is at least $l_m+d+2$. Hence
    $[l_m,l_m+d+1]$ carries at most one phrase-start. Now let $p\in[l_m+d+2,r_m]$ be a phrase-start with
    $p\le r_m-1$; then $\mathsf{src}_m(p)\ge r'_j+2$, so
    $\sigma\cap[\mathsf{src}_m(p),r'_m]\subseteq\sigma\cap(r'_j,r'_m]=\emptyset$ and
    \eqref{eq:clean-copy} applied to $(p,r_m-1)$ shows that the phrase starting at $p$ ends at position
  at least $r_m$. Hence $[l_m+d+2,r_m]$ carries at most one phrase-start as well.} So $s_m\le2$, a
  contradiction.

  \smallskip\noindent\emph{Item~\ref{it:clean-right}.}
  Assume first $j<m$. Then $w'[l_m:r_m-1]=w[l'_m:r'_m]$ is a suffix of $w[l'_j:r'_j]=w'[l_j:r_j-1]$,
  occurring at $l_j+(l'_m-l'_j)\le r_j-1<l_m$; as in item~(\ref{it:clean-inj}) this gives $s_m\le1$, a
  contradiction. Hence $m<j$, and $r_m<l_j$.

  Suppose now that $\sigma\cap[l'_j,l'_m)=\emptyset$ and set $e=l'_m-l'_j\ge1$, so that
  $l_j+e-1\le r_j-2$. \corrindex{Again we split $[l_j,r_j]$, this time at $l_j+e$, and again $l_j$ need
    not be a phrase-start of $w'$. Let $p\in[l_j,l_j+e-1]$ be a phrase-start; then
    $\sigma\cap[\mathsf{src}_j(p),l'_m-1]\subseteq\sigma\cap[l'_j,l'_m)=\emptyset$, so
    \eqref{eq:clean-copy} applied to $(p,l_j+e-1)$ shows that the phrase starting at $p$ ends at
    position at least $l_j+e$. Hence $[l_j,l_j+e-1]$ carries at most one phrase-start. Now let
  $p\in[l_j+e,r_j]$ be a phrase-start with $p\le r_j-1$, so that $\mathsf{src}_j(p)\ge l'_m$.}
  Using that $P_j$ and $P_m$ are clean and that
  $r'_j=r'_m$,
  \[
    w'[p:r_j-1]=w[\mathsf{src}_j(p):r'_j]=w[\mathsf{src}_j(p):r'_m]
    = w'\bigl[l_m+(\mathsf{src}_j(p)-l'_m) : r_m-1\bigr] ,
  \]
  an occurrence starting at position at most $r_m-1<l_j\le p$. Hence the phrase starting at $p$ ends at
  position at least $r_j$, \corrindex{so $[l_j+e,r_j]$ carries at most one phrase-start as well}. So
  $s_j\le2$, a contradiction.
\end{proof}

\section{Properties of P-systems}\label{app:p-systems}

We now prove propositions \ref{prop:p2} and \ref{prop:p3}
\begin{proof}[Proof of \cref{prop:p2}]

  Since $(E,\sigma)$ is a P2-system, for every $(a,b)\in E$ we can associate an
  inner marker
  \[
    l_{(a,b)}\in\sigma
    \quad\text{with}\quad
    a\le l_{(a,b)}\corrindex{\le} b .
  \]
  \corrindex{We write $\sigma_1<\sigma_2<\cdots<\sigma_{\card\sigma}$ for the elements of $\sigma$ in
    increasing order, so that $\sigma_i$ is the \emph{position} of the $i$-th marker, and for
  $1\le i\le\card\sigma$ we define}
  \[
    E_i=\{e\in E : l_e=\corrindex{\sigma_i}\},
    \qquad
    n_i=\card{E_i}
    \qquad\text{and}\qquad
    \corrindex{\|E_i\|}=\sum_{(a, b)\in E_i}(b-a) .
  \]

  First, for $j\ge 0$, let $n_{i,j}$ be the number of elements in $E_i$ for which the left bracket is at distance $j$ from the marker:
  \[n_{i,j}=\card{\{(a,b)\in E_i : a=\corrindex{\sigma_i}-j\}}.\]

  \begin{figure}[ht]
    \centering
    \begin{tikzpicture}[x=0.48cm, y=0.48cm, >=Stealth]
      \draw[editred, dashed, thick] (5,-1.2) -- (5,6.2);
      \fill[editred] (5,-1.2) circle (3pt);
      \node[editred, font=\small\bfseries, below=3pt] at (5,-1.2) {$\sigma_i$};

      \draw[pipeblue, very thick, line cap=round] (5,0.0) -- (6.5,0.0);
      \fill[pipeblue] (6.5,0.0) circle (2pt);
      \draw[pipeblue, very thick, line cap=round] (5,0.8) -- (8,0.8);
      \fill[pipeblue] (8,0.8) circle (2pt);
      \draw[pipeblue, very thick, line cap=round] (5,1.6) -- (10,1.6);
      \fill[pipeblue] (10,1.6) circle (2pt);
      \node[pipeblue, font=\small\bfseries, anchor=east] at (4.5,0.8) {$j\!=\!0$};
      \draw[decorate, decoration={brace, amplitude=3pt, mirror}, pipeblue!70]
      (10.6,-0.2) -- (10.6,1.8)
      node[midway, right=4pt, font=\small, pipeblue]{$n_{i,0}\!=\!3$};

      \draw[c4, very thick, line cap=round] (4,2.8) -- (7,2.8);
      \fill[c4] (7,2.8) circle (2pt);
      \draw[c4, very thick, line cap=round] (4,3.6) -- (11,3.6);
      \fill[c4] (11,3.6) circle (2pt);
      \node[c4, font=\small\bfseries, anchor=east] at (3.5,3.2) {$j\!=\!1$};
      \draw[decorate, decoration={brace, amplitude=3pt, mirror}, c4!70]
      (11.6,2.6) -- (11.6,3.8)
      node[midway, right=4pt, font=\small, c4]{$n_{i,1}\!=\!2$};

      \draw[c3, very thick, line cap=round] (3,5.4) -- (12,5.4);
      \fill[c3] (12,5.4) circle (2pt);
      \node[c3, font=\small\bfseries, anchor=east] at (2.5,5.4) {$j\!=\!2$};
      \node[c3, font=\small, anchor=west] at (12.5,5.4) {$n_{i,2}\!=\!1$};

      \node[font=\small, text width=3.8cm, align=left, anchor=west,
      draw=gray!40, rounded corners=2pt, fill=gray!5, inner sep=4pt]
      at (15.5,2.8)
      {Same $j$: same left end\\$\Rightarrow$ right ends distinct\\
      $\Rightarrow$ brackets spread out};
    \end{tikzpicture}
    \caption{Grouping brackets in class $E_i$ by distance $j$ from the marker $\sigma_i$. Brackets sharing the same left endpoint must have distinct right endpoints, so they fan out.}
    \label{fig:p2-proof-grouping}
  \end{figure}

  Observe that by definition we have $(\dagger)$: $\sum_{j\ge0}n_{i,j}=n_i$.
  Fix $j\ge0$. The $n_{i,j}$ brackets of $E_i$ with left endpoint $\sigma_i-j$ are pairwise distinct and
  share their left endpoint, hence they have pairwise distinct right endpoints; moreover each such right
  endpoint $b$ satisfies $b\ge \sigma_i$, since the marker $\sigma_i$ lies in $[\sigma_i-j,b]$. The
  quantities $b-\sigma_i$ are therefore $n_{i,j}$ pairwise distinct non-negative integers, so
  \corrindex{$\sum_{b : (\sigma_i-j,b)\in E_i}(b-\sigma_i)\ge 0+1+\cdots+(n_{i,j}-1) = \frac{n_{i,j}^2}{2}-\frac{n_{i,j}}{2}$}. Hence:
  \begin{align*}
    \|E_i\|
    &=
    \sum_{(a, b)\in E_i}(b-a)\\
    &=
    \sum_{j\ge0}\sum_{b:(\sigma_i-j, b)\in E_i}(j+b-\sigma_i)\\
    &\ge
    \sum_{j\ge0} n_{i,j}\left(\frac{n_{i,j}}2+j\right)\corrindex{-\frac{n_i}{2}} .
  \end{align*}
  \corrindex{The linear correction $-n_i/2$ is exactly the price of allowing a marker to sit on the right
    endpoint of a bracket: when $b>\sigma_i$ is required for every bracket, the distinct values
  $b-\sigma_i$ are positive and the correction disappears.}

  Combining with \cref{cl:opt_p2}, which applies thanks to $(\dagger)$, yields
  \[
    \|E_i\| \ge \frac{n_i^{3/2}}6\corrindex{-\frac{n_i}{2}}.
  \]

  Finally, since the classes $E_i$ partition $E$,
  \[
    \|E\|
    =
    \sum_{i=1}^{\card\sigma}\|E_i\|
    \ge
    \frac 1 6 \sum_{i=1}^{\card\sigma} n_i^{3/2}
    \corrindex{-\frac{\card E}{2}} .
  \]
  Apply Hölder's inequality with exponents $p=3/2$ and $q=3$:
  \[
    \sum_{i=1}^{\card\sigma} n_i^{3/2}
    \ge
    \card{\sigma}^{-1/2}\Bigl(\sum_{i=1}^{\card\sigma} n_i\Bigr)^{3/2}
    =
    \card{\sigma}^{-1/2} \card{E}^{3/2}.
  \]
  Therefore
  \[
    \|E\| \ge \frac{\card{E}^{3/2}}{6\sqrt{\card{\sigma}}}\corrindex{-\frac{\card{E}}{2}}
    \ge \frac{\card{E}^{3/2}}{8\sqrt{\card{\sigma}}}-\frac{\card{E}}{2} ,
  \]
  as wanted.
\end{proof}

\begin{claim}
  \label{cl:opt_p2}
  For $(n_j)_{j\in\nn}\in\nn^\nn$,

  \[\text{if }\sum_{j\ge0}n_j = n\quad\text{ then }\sum_{j\ge0} n_{j}\left(\frac{n_{j}}2+j\right)\ge \frac{n^{3/2}}{\corrindex{6}} .\]
\end{claim}

\begin{claimproof}
  For $g : j \mapsto   n_{j}$, let
  \[f(g) = \sum_{j\ge0} g(j)\left(\frac{g(j)}2+j\right)\]
  Let $g$ be the function such that $\sum_{j\ge0}g(j) = n$ and $f(g)$ is minimal. Let $j>2\sqrt{n}$. Suppose $g(j)\neq0$. Then, by the pigeonhole principle, there exists $j'<\sqrt{n}$ such that $g({j'})<\sqrt{n}$. Define
  $$
  g'(i)=
  \begin{cases}
    g(i)+1 \text{ if } i = j',\\
    g(i)-1 \text{ if } i = j,\\
    g(i) \text{ otherwise.}
  \end{cases}
  $$
  Then,
  \begin{align*}
    f(g')-f(g)
    &=
    (g(j')+1)\left(\frac{g({j'})+1}2+j'\right)-g({j'})\left(\frac{g({j'})}2+j'\right)\\
    &+(g(j)-1)\left(\frac{g(j)-1}2+j\right)-g({j})\left(\frac{g({j})}2+j\right)\\
    &=
    \left(\frac{g({j'})+1}2+j'\right) + \frac{g(j')}2 -\left(g({j}) + j -\frac{1}{2}\right) \\
    &<
    \left(\frac{\sqrt n}2+\sqrt n\right) + \frac{\sqrt n}2 - 2\sqrt n\qquad\text{(by definition of $j$ and $j'$)}\\
    &< 0
  \end{align*}
  This contradicts the minimality of $f(g)$. Hence, $(\ddagger)$: for $j>2\sqrt{n}$, $g(j)=0$; in
  particular $\sum_{j=0}^{\lfloor 2\sqrt n\rfloor} g(j)=n$. To conclude, for $(n_j)_{j\in\nn}\in\nn^\nn$
  such that $\sum_{j\ge0}n_j = n$,
  \begin{align*}
    \sum_{j\ge0} n_j\left(\frac{n_j}2+j\right)
    &\ge
    \sum_{j\ge0} g(j)\left(\frac{g(j)}2+j\right)
    &(\text{by minimality of }g)\\
    &\ge
    \frac12\sum_{j=0}^{\lfloor 2\sqrt{n}\rfloor}g({j})^2
    &(\text{by }(\ddagger))\\
    &\ge
    \frac12\cdot\frac{\left(\sum_{j=0}^{\lfloor 2\sqrt{n}\rfloor} g({j})\right)^2}{2\sqrt{n}+1}
    &(\text{by Cauchy--Schwarz})\\
    &=
    \frac{n^2}{2(2\sqrt n+1)}
    \ge \frac{n^{3/2}}6
    &(\text{because } 2\sqrt n+1\le 3\sqrt n \text{ for } n\ge1)
  \end{align*}
\end{claimproof}

\begin{proof}[Proof of \cref{prop:p3}]
  Since $(E,\sigma)$ is a P3-system, for every $(a,b)\in E$ we can associate its first inner marker, denoted by $i_{(a,b)}$. In symbols:
  \[
    i_{(a,b)}\in\sigma
    \quad\text{with}\quad
    a\le i_{(a,b)}\corrindex{\le} b
    \text{ and }
    \text{for }
    i\in\sigma,
    a\le i \Rightarrow i_{(a,b)} \le i .
  \]
  We write $\sigma_1<\sigma_2<\cdots<\sigma_k$ for the elements of $\sigma$ in increasing order,
  with the conventions $\sigma_0=-\infty$ and $\sigma_{k+1}=+\infty$,
  and we index the classes below by the rank of the marker.
  For each $i\in\sigma$, we define
  \[
    E_i=\{e\in E : i_e=i\}
    \qquad\text{and}\qquad
    n_i=\card{\{e\in E : i_e=i\}}.
  \]
  Note that $n_k=0$ because of condition (i) in the definition of \emph{P3-systems}. This value can thus be ignored, and attention is restricted to the first $k-1$ values of $n_i$.

  Observe that for $(a, b)\in E$, the set $\{c\in\mathsf{dom}(E) : a<c<b\}$ consists of distinct integers
  lying strictly between $a$ and $b$, of which there are exactly $b-a-1$; hence
  \[(b-a)\ge \card{\{c\in \mathsf{dom}(E), a< c< b\}}+1\]
  where  $\mathsf{dom}(E) = \{a: \exists b, (a, b)\in E \text{ or } (b, a)\in E\}$.

  It follows that a lower bound on the number of endpoints strictly inside $[a, b]$ gives a lower bound on $(b-a)$. The goal is to give a lower bound of $\|E\|$ based on the $n_i$. Write
  \[
    L_i=\{a : (a,b)\in E_i\}
    \qquad\text{and}\qquad
    R_i=\{b : (a,b)\in E_i\} ,
  \]
  which both have exactly $n_i$ elements by condition (ii). We record three facts.
  \begin{enumerate}
    \item[\textup{(a)}] $L_i\subseteq(\sigma_{i-1},\sigma_i]$. Indeed $a\le\sigma_i$ because $\sigma_i$ is a marker of the bracket, and $a>\sigma_{i-1}$ because $\sigma_i$ is its \emph{first} marker.
    \item[\textup{(b)}] $R_i\subseteq[\sigma_{i+1},+\infty)$. Indeed condition (i) provides a second marker $j$ of the bracket with $\sigma_i<j\le b$, and any such $j$ is at least $\sigma_{i+1}$.
    \item[\textup{(c)}] Consequently $L_i$, $R_i$ and $L_{i+1}$ are pairwise disjoint, \corrindex{with the single possible exception of $R_i\cap L_{i+1}\subseteq\{\sigma_{i+1}\}$}: by (a) and (b) we have $L_i\le\sigma_i<\sigma_{i+1}\le R_i$ and $\sigma_i<L_{i+1}\le\sigma_{i+1}$.
  \end{enumerate}

  Two effects are then taken into consideration (see Figure~\ref{fig:p3-proof-contributions}):
  \begin{itemize}
    \item The effect of the brackets associated with a same marker. For $(a,b)\in E_i$, the endpoints of $E_i$ lying strictly inside $(a,b)$ are the $a'\in L_i$ with $a'>a$ (they satisfy $a'\le\sigma_i<\sigma_{i+1}\le b$) and the $b'\in R_i$ with $b'<b$ (they satisfy $b'\ge\sigma_{i+1}>\sigma_i\ge a$). Since by condition (ii) the left endpoints are pairwise distinct and so are the right endpoints, summing over $(a,b)\in E_i$ gives \corrindex{$2\binom{n_i}2=n_i^2-n_i$}, and \emph{not} $n_i^2$;
    \item The effect of the brackets associated with the next marker. For $(a',b')\in E_{i+1}$ and $(a,b)\in E_i$, fact (c) gives $a<\sigma_i<a'\le\sigma_{i+1}\le b$, so $a'$ lies strictly inside $(a,b)$ \corrindex{unless $a'=b=\sigma_{i+1}$}. This sums up to \corrindex{at least $n_in_{i+1}-\delta_i$, where $\delta_i\in\{0,1\}$}, since at most one bracket of $E_i$ has right endpoint $\sigma_{i+1}$. Note that as $n_k=0$, this term does not appear for $i=k-1$.
  \end{itemize}

  \corrindex{Finally, the two effects may count one common point twice, namely $\sigma_{i+1}$ when it belongs to
    $R_i\cap L_{i+1}$. If it does not, then $\delta_i=0$ and there is no double count. If it does, then
    $\sigma_{i+1}$ is counted twice exactly for the $n_i-1$ brackets of $E_i$ whose right endpoint is larger
    than $\sigma_{i+1}$, while $\delta_i=1$ accounts for the remaining one. In both cases the total loss over
    the class $E_i$ is at most $n_i$, which is exactly compensated by the $n_i$ units of slack provided by
  the ``$+1$'' above. Altogether,}
  \[
    \sum_{(a,b)\in E_i}(b-a)
    \ge
    \underbrace{(n_i^2-n_i)}_{\text{same marker}}
    +\underbrace{n_in_{i+1}}_{\text{next marker}}
    +\underbrace{n_i}_{\text{the ``}+1\text{''}}
    -\underbrace{n_i}_{\text{double count}}
    =
    n_i^2+n_in_{i+1}-n_i .
  \]

  \begin{figure}[ht]
    \centering
    \begin{tikzpicture}[x=0.42cm, y=0.42cm, >=Stealth,scale=0.95]
      \draw[editred, dashed, thick] (4,-1.5) -- (4,7.5);
      \fill[editred] (4,-1.5) circle (3pt);
      \node[editred, font=\small\bfseries, below=3pt] at (4,-1.5) {$\sigma_i$};

      \draw[editred, dashed, thick] (10,-1.5) -- (10,7.5);
      \fill[editred] (10,-1.5) circle (3pt);
      \node[editred, font=\small\bfseries, below=3pt] at (10,-1.5)
      {$\sigma_{i+1}$};

      \draw[pipeblue, very thick, line cap=round] (0,0.8) -- (15,0.8);
      \fill[pipeblue] (0,0.8) circle (2pt); \fill[pipeblue] (15,0.8) circle (2pt);
      \draw[pipeblue, very thick, line cap=round] (1.5,2.0) -- (13.5,2.0);
      \fill[pipeblue] (1.5,2.0) circle (2pt); \fill[pipeblue] (13.5,2.0) circle (2pt);
      \draw[pipeblue, very thick, line cap=round] (3,3.2) -- (12,3.2);
      \fill[pipeblue] (3,3.2) circle (2pt); \fill[pipeblue] (12,3.2) circle (2pt);

      \node[pipeblue, font=\small\bfseries] at (7.5,-0.2) {$E_i$};

      \draw[decorate, decoration={brace, amplitude=4pt}, pipeblue!70]
      (-1.2,0.5) -- (-1.2,3.5)
      node[midway, left=5pt, font=\small, pipeblue, text width=2.2cm, align=right]
      {distinct endpoints\\ $\Rightarrow\; n_i^2-n_i$};

      \draw[c4, very thick, line cap=round] (6,5.4) -- (18,5.4);
      \fill[c4] (6,5.4) circle (2pt); \fill[c4] (18,5.4) circle (2pt);
      \draw[c4, very thick, line cap=round] (8,6.6) -- (17,6.6);
      \fill[c4] (8,6.6) circle (2pt); \fill[c4] (17,6.6) circle (2pt);

      \node[c4, font=\small\bfseries] at (7,7.4) {$E_{i+1}$};

      \draw[c4, -{Stealth}, semithick, opacity=0.6] (6,5.0) -- (6,3.5);
      \draw[c4, -{Stealth}, semithick, opacity=0.6] (8,5.0) -- (8,3.5);

      \node[c4, font=\small, anchor=west, text width=3cm, align=left]
      at (18.8,4.8)
      {left endpoints of $E_{i+1}$\\fall inside all\\$E_i$ brackets\\
      $\Rightarrow\; n_i \cdot n_{i+1}$};
    \end{tikzpicture}
    \caption{The two contributions to the extent lower bound: same-class brackets with distinct endpoints contribute $n_i^2-n_i$, while cross-class brackets from $E_{i+1}$ nested inside $E_i$ brackets contribute $n_i \cdot n_{i+1}$. The two contributions overlap in at most one point, $\sigma_{i+1}$.}
    \label{fig:p3-proof-contributions}
  \end{figure}

  In the end, summing over the classes and using $\sum_i n_i=\card E$:
  \[\|E\|\ge\sum_{i=1}^{k-2}\left(n_in_{i+1}+{n_i^2}\right)+{n_{k-1}^2}\corrindex{-\card{E}},\]

  \noindent The new goal is to provide a lower bound on the objective function
  \[F(n_1,\dots,n_{k-1})=\sum_{i=1}^{k-2}\left(n_in_{i+1}+{n_i^2}\right)+{n_{k-1}^2}\]
  under the constraint $\sum_{i=1}^{k-1}n_i=n$, where $n=\card{E}$. Setting $n_k=0$ and
  $n_{k+1}=n_1$, the objective becomes cyclic,
  \[
    F(n_1,\dots,n_{k-1})
    = \sum_{i=1}^{k}\left(n_in_{i+1}+n_i^2\right)
    = \frac12\sum_{i=1}^{k}\left(n_i+n_{i+1}\right)^2 ,
  \]
  where the first equality holds because the two extra terms $n_k^2$ and $n_kn_1$ vanish, and the second
  one because $\sum_{i=1}^k n_{i+1}^2=\sum_{i=1}^k n_i^2$ by cyclicity. By the Cauchy--Schwarz
  inequality,
  \[
    \frac12\sum_{i=1}^{k}\left(n_i+n_{i+1}\right)^2
    \ge
    \frac{1}{2k}\left(\sum_{i=1}^{k}\left(n_i+n_{i+1}\right)\right)^2
    =
    \frac{(2n)^2}{2k}
    =
    2\frac{n^2}{k} .
  \]
  Hence, with $n=\card E$ and $k=\card\sigma$,
  \[\|E\| \ge 2\frac{\card{E}^2}{\card\sigma}-\card{E},\]
  as wanted.
\end{proof}

\subsection{Lower bound from compression ratio and few edits}
\label{app:lb-two}

The goal of this section is to prove the following lower bounds that complement the trichotomy theorem.
\TrichotomyLower*

\subsubsection{First point of \cref{prop:trichotomy-lower}}
The idea is to build a word $w$ whose \LZ parse is dominated by a large table of distinct short blocks that are all
copied from an initial header, so each table cell costs one phrase. A single-symbol change
in the header then destroys these long matches and forces almost all cells to be parsed into
two much shorter phrases, yielding an almost factor-$2$ blow-up.

\paragraph{The construction.}
Let $k\ge1$ be the number of edits and let $M\ge 1$ be a parameter, to be chosen large enough at the end of the proof.
For $1\le \ell\le k$, let all symbols
\[
  \{\#_{i,j,\ell}\colon 0\le i,j\le M\}\cup\{a_\ell,b_\ell,c_\ell,\#_\ell\}
\]
be distinct and also disjoint across different $\ell$'s. Define
\[
  u_{i,j,\ell}=a_\ell^i\,b_\ell\,a_\ell^j\,\#_{i,j,\ell},\qquad
\]\[
  u_\ell=a_\ell^M\,b_\ell\,a_\ell^M\,\#_\ell\cdot\prod_{i=0}^{M}\prod_{j=0}^{M}u_{i,j,\ell}\qquad\text{and}\qquad
  u'_\ell=a_\ell^M\,c_\ell\,a_\ell^M\,\#_\ell\cdot\prod_{i=0}^{M}\prod_{j=0}^{M}u_{i,j,\ell}.
\]
Finally set
\[
  u=\prod_{\ell=1}^{k}u_\ell\qquad\text{and}\qquad
  u'=\prod_{\ell=1}^{k}u'_\ell.
\]
Observe that $u'$ is obtained from $u$ by exactly $k$ edits (one substitution $b_\ell\to c_\ell$ per block).

\paragraph{The parses.}
Within each block $\ell$, the header $a_\ell^M b_\ell a_\ell^M \#_\ell$ makes every
\corrindex{$a_\ell^i b_\ell a_\ell^j$ (for $0\le i,j\le M$) appear earlier as an exact copy}, hence
each $u_{i,j,\ell}=a_\ell^i b_\ell a_\ell^j \#_{i,j,\ell}$ contributes one phrase to $LZ77(u)$, namely
the copy of $a_\ell^i b_\ell a_\ell^j$ followed by the fresh symbol $\#_{i,j,\ell}$. Replacing $b_\ell$ by $c_\ell$
in the header removes that exact earlier occurrence; each $u_{i,j,\ell}$ must then be parsed
as two phrases.
Formally, for every $1\le \ell\le k$,
\begin{align*}
  LZ77(u_\ell)
  =& LZ77(a_\ell^M b_\ell a_\ell^M \#_\ell)
  \mid u_{0,0,\ell} \mid u_{0,1,\ell} \mid \cdots \mid u_{M,M,\ell},\\[2mm]
  LZ77(u'_\ell)
  =& LZ77(a_\ell^M c_\ell a_\ell^M \#_\ell)\\
  &\mid b_\ell \mid \#_{0, 0, \ell}
  \mid b_\ell a_\ell \mid \#_{0, 1, \ell}
  \mid \cdots
  \mid b_\ell a_\ell^M \mid \#_{0, M, \ell}\\
  &\mid a_\ell b_\ell \mid \#_{1, 0, \ell}
  \mid a_\ell b_\ell a_\ell \mid \#_{1, 1, \ell}
  \mid \cdots
  \mid a_\ell b_\ell a_\ell^M \mid \#_{1, M, \ell}\\
  &\mid \cdots\cdots\\
  &\mid a_\ell^Mb_\ell \mid \#_{M, 0, \ell}
  \mid a_\ell^Mba_\ell \mid \#_{M, 1, \ell}
  \mid \cdots
  \mid a_\ell^M b_\ell a_\ell^M \mid \#_{M, M, \ell}\\
\end{align*}
Because the alphabets are disjoint across $\ell$ and each $u_i$ ends with a complete block,
the concatenations preserve parse independence.
\[
  LZ77(u)=LZ77(u_1)\mid LZ77(u_2)\mid\cdots\mid LZ77(u_k)
  \text{ and }
  C(u)=\sum_{\ell=1}^k C(u_\ell) ,
\]
and the same holds for $u'$.

Observe also that, \corrindex{using self-reference},
\[
  \corrindex{LZ77(a_\ell^M b_\ell a_\ell^M \#_\ell) = a_\ell \mid a_\ell^{M-1}b_\ell \mid a_\ell^M \#_\ell ,
    \text{ so }
  C(a_\ell^M b_\ell a_\ell^M \#_\ell)=C(a_\ell^M c_\ell a_\ell^M \#_\ell)=3 ,}
\]
so the headers have negligible cost.

\paragraph{Counting length and number of phrases.}
For each $\ell$,
\[
  |u_\ell|
  = |a_\ell^M b_\ell a_\ell^M \#_\ell|
  + \sum_{i=0}^{M}\sum_{j=0}^{M}|u_{i,j,\ell}|
  = (2M+2)+\sum_{i=0}^{M}\sum_{j=0}^{M}(2+i+j)
  \underset{M\shortrightarrow\infty}{\sim} M^3.
\]
Hence
\[
  |u|=|u'|
  \underset{M\shortrightarrow\infty}{\sim} k M^3 .
\]
For the parse sizes we get, using the displays above,
\[
  C(u)\underset{M\shortrightarrow\infty}{\sim} k M^2
  \text{ and }
  C(u')\underset{M\shortrightarrow\infty}{\sim} 2 k M^2 .
\]
hence
\[
  \frac{C(u')}{C(u)}\xrightarrow[M\to\infty]{} 2.
\]
Equivalently, for every fixed $0<\delta<1$ there exists $M_0(\delta)$ such that for all $M\ge M_0(\delta)$,
\begin{equation}
  \label{eq:r-two-eta}
  \frac{C(u')}{C(u)}\ge 2-\delta.
\end{equation}
Finally,
\[
  C(u)\ge k(M+1)^2\underset{M\shortrightarrow\infty}{\sim}  kM^2\underset{M\shortrightarrow\infty}{\sim} k^{1/3}\,|u|^{2/3},
\]
as desired.

\subsubsection{Second point of \cref{prop:trichotomy-lower}}
We now prove the second point of \cref{prop:trichotomy-lower} by constructing a second family exhibiting an almost factor-3 blow-up.

\paragraph{Construction.}
Fix integers $p,k'\ge1$.
For each $1\le \ell\le k'$, let the symbols $a_\ell,b_\ell,c_\ell$ be distinct and unique across $\ell$.
Define
\[
  u_{\ell, i}=a_\ell^i\,b_\ell b_\ell\,a_\ell^{i+1},
  \qquad
  v=\Bigl(\prod_{\ell=1}^{k'}u_{\ell, p}\Bigr)\#
  \text{ and }
  v'=\Bigl(\prod_{\ell=1}^{k'}a_\ell^p\,c_\ell c_\ell\,a_\ell^{p+1}\Bigr)\#,
\]
so $v'$ differs from $v$ by exactly $k=2k'$ substitutions ($b_\ell\!\to c_\ell$ twice per factor). \corrindex{In particular $k$ is even, which is why the second statement of \cref{prop:trichotomy-lower} is restricted to even $k$: two substitutions in the same block are what splits a phrase into three, whereas a single one only splits it into two.}

We group the $u_{\ell, i}$ as
\[
  u_{\ell} = \prod_{i=0}^{p} u_{\ell,i}\#_{\ell, i}
\]
and we define
\[
  u = v \cdot \prod_{\ell=1}^{k'} u_{\ell} \qquad \text{and}\qquad
  u' = v' \cdot \prod_{\ell=1}^{k'} u_{\ell}.
\]
Each $u_{\ell, i}$ family uses disjoint alphabets and fresh separators $\#_{\ell,i}$.

\paragraph{\LZ parses.}
In $u$, the prefix $v$ ensures that every
$u_{\ell,i}$ reoccurs as a single match:
each is copied as one phrase after the phrase covering $v$.
Thus
\begin{align*}
  LZ77(u) =
  & LZ77(v)\\
  &\mid u_{1, 0} \mid \cdots \mid u_{1, p}  \\
  &\mid u_{2, 0} \mid \cdots \mid u_{2, p}  \\
  &\mid \cdots\cdots \\
  &\mid u_{k', 0} \mid \cdots \mid u_{k', p}\\
\end{align*}
while the parsing of $u'$ looks like this
\begin{align*}
  LZ77(u') =
  & LZ77(v')\\
  &\mid b_1 \mid b_1 a_1 \mid \#_{1, 0} \mid \cdots \mid a_1^p b_1 \mid b_1 a_1^{p+1} \mid \#_{1,p} \\
  &\mid b_2 \mid b_2 a_2 \mid \#_{2, 0} \mid \cdots \mid a_2^p b_2 \mid b_2 a_2^{p+1} \mid \#_{2,p} \\
  &\mid \cdots\cdots \\
  &\mid b_{k'} \mid b_{k'} a_{k'} \mid \#_{k', 0} \mid \cdots \mid a_{k'}^p b_{k'} \mid b_{k'} a_{k'}^{p+1} \mid \#_{k',p} \\
\end{align*}

\paragraph{Counting length and phrases.}
We compute
\[
  |u| = |u'| = \sum_{\ell=1}^{k'}(2p+3) \corrindex{+\;1} + \sum_{\ell=1}^{k'}\sum_{i=0}^{p}(2i+4)
  \underset{p\shortrightarrow\infty}{\sim} p^2k' ,
\]
\corrindex{where the $+1$ accounts for the final $\#$ of $v$ (resp. $v'$)}.
Observe also that the header parses explicitly as
\[
  LZ77(v) = a_1 \mid a_1^{p-1}b_1 \mid b_1a_1 \mid a_1^{p}a_2 \mid a_2^{p-1}b_2 \mid \cdots
  \mid a_{k'}^{p-1}b_{k'} \mid b_{k'}a_{k'} \mid a_{k'}^{p}\# ,
\]
which uses $4$ phrases for the first factor and $3$ for each of the $k'-1$ others, so that
\[C(v)=C(v')=3k'+1 = \underset{p\rightarrow\infty}{o}(p k') ,\]
so the headers will have negligible cost in front of the rest of the parsing.

For the parse sizes we get, using the displays above and the fact that $i$ ranges over $p+1$ values,
\[
  \corrindex{C(u) = C(v) + k'(p+1)} \underset{p\shortrightarrow\infty}{\sim} pk'
  \quad\text{and}\quad
  \corrindex{C(u') = C(v') + 3k'(p+1)} \underset{p\shortrightarrow\infty}{\sim} 3pk' .
\]
Hence, for every $\delta>0$ there exists a large enough $p$ such that
\[
  \frac{C(u')}{C(u)}\ge 3-\delta.
\]
Finally, observe that
\begin{align*}
  C(u) &\underset{p\shortrightarrow\infty}{\sim} pk'
  \underset{p\shortrightarrow\infty}{\sim} k'\sqrt{|u|/k'}
  \underset{p\shortrightarrow\infty}{\sim}\sqrt{\frac{k|u|}2}\\
\end{align*}
as wanted.

\end{document}